\documentclass[a4paper,fleqn]{cas-sc}

\usepackage[authoryear,longnamesfirst]{natbib}

\usepackage{amsmath,amsthm}
\usepackage{amsfonts}
\usepackage{comment}
\usepackage{color}
\newtheorem{theorem}{Theorem}[section]
\newtheorem{lemma}{Lemma}[section]
\newtheorem{corollary}{Corollary}[section]
\newtheorem{proposition}{Proposition}[section]
\theoremstyle{definition}

\def\tsc#1{\csdef{#1}{\textsc{\lowercase{#1}}\xspace}}
\tsc{WGM}
\tsc{QE}
\begin{document}
\let\WriteBookmarks\relax
\def\floatpagepagefraction{1}
\def\textpagefraction{.001}

% Short title
\shorttitle{Constitutive Modelling of Korteweg Fluids}    

% Short author
\shortauthors{Z. Mati\'{c}, S. Simi\'{c} and P. V\'{a}n}  

% Main title of the paper
\title [mode = title]{Constitutive Modelling of Korteweg Fluids Using Liu's Method}  

% Title footnote mark
% eg: \tnotemark[1]
%\tnotemark[1] 

% Title footnote 1.
% eg: \tnotetext[1]{Title footnote text}
%\tnotetext[1]{This is footnote text.} 

% First author
%
% Options: Use if required
% eg: \author[1,3]{Author Name}[type=editor,
%       style=chinese,
%       auid=000,
%       bioid=1,
%       prefix=Sir,
%       orcid=0000-0000-0000-0000,
%       facebook=<facebook id>,
%       twitter=<twitter id>,
%       linkedin=<linkedin id>,
%       gplus=<gplus id>]

\author[1]{Zagorka Mati\'{c}}[orcid=0000-0001-8586-733X]

% Footnote of the first author
%\fnmark[1]

% Email id of the first author
\ead{zagorka.mat@uns.ac.rs}

% URL of the first author
%\ead[url]{}

% Credit authorship
% eg: \credit{Conceptualization of this study, Methodology, Software}
%\credit{}

% Address/affiliation
\affiliation[1]{organization={Faculty of Technical Sciences, University of Novi Sad},
            addressline={Trg Dositeja Obradovi\'{c}a 6}, 
            city={Novi Sad},
            citysep={}, % Uncomment if no comma needed between city and postcode
            postcode={21102}, 
%           state={},
            country={Serbia}}

\author[2]{Srboljub Simi\'{c}}[orcid=0000-0003-3726-2007]

% Corresponding author indication
\cormark[1]

% Footnote of the second author
%\fnmark[2]

% Email id of the second author
\ead{ssimic@uns.ac.rs}

% URL of the second author
%\ead[url]{}

% Credit authorship
%\credit{}

% Address/affiliation
\affiliation[2]{organization={Faculty of Sciences, University of Novi Sad},
            addressline={Trg Dositeja Obradovi\'{c}a 4}, 
            city={Novi Sad},
            citysep={}, % Uncomment if no comma needed between city and postcode
            postcode={21000}, 
%           state={},
            country={Serbia}}

% Corresponding author text
\cortext[1]{Corresponding author}

\author[3,4]{Peter V\'{a}n}[orcid=0000-0002-9396-4073]

% Corresponding author indication
%\cormark[1]

% Footnote of the second author
%\fnmark[2]

% Email id of the second author
\ead{van.peter@wigner.hun-ren.hu}

% URL of the second author
%\ead[url]{}

% Credit authorship
%\credit{}

% Address/affiliation
\affiliation[3]{organization={Department of Theoretical Physics, Institute of Particle and Nuclear Physics, HUN-REN Wigner Research Centre for Physics},
		addressline={Konkoly-Thege Mikl\'{o}s \'{u}t 29-33.}, 
		city={Budapest},
	    citysep={}, % Uncomment if no comma needed between city and postcode
		postcode={H-1121}, 
%   	state={},
		country={Hungary}}

% Address/affiliation
\affiliation[4]{organization={Department of Energy Engineering, Faculty of Mechanical Engineering, Budapest University of Technology and Economics},
		addressline={M\H{u}egyetem rkp. 3.}, 
		city={Budapest},
   		citysep={}, % Uncomment if no comma needed between city and postcode
		postcode={H-1111}, 
%   	state={},
		country={Hungary}}

% Corresponding author text
%\cortext[1]{Corresponding author}

% Footnote text
%\fntext[1]{}

% For a title note without a number/mark
%\nonumnote{}

% Here goes the abstract
\begin{abstract}
The paper studies constitutive modelling of Korteweg fluids. Thermodynamic consistency, i.e. compatibility with entropy balance law, is achieved using Liu's method of multipliers. Appropriate constitutive assumptions facilitated inclusion of the capillary effects in the specific entropy. Korteweg stresses are derived from the equilibrium conditions---vanishing of the entropy production and its minimization in equilibrium. Material parameter in Korteweg stresses is allowed to depend on temperature, which turns out to be consistent with kinetic-theory results and leads to cross-coupling of mechanical and thermal effects. The generalized Gibbs' relation, which inherits the capillary effects, is derived as consequence, which is a peculiar feature of the Liu's method. 
\end{abstract}

% Use if graphical abstract is present
%\begin{graphicalabstract}
%\includegraphics{}
%\end{graphicalabstract}

% Research highlights
%\begin{highlights}
%\item 
%\item 
%\item 
%\end{highlights}

% Keywords
% Each keyword is seperated by \sep

\maketitle

% Main text

%%%%%% 

\section{Introduction} 

Capillarity is the phenomenon which occurs in fluids when they are in contact with solid material in a narrow region. It is a consequence of the combined effects of adhesion and surface tension. One of the first attempts to capture these effects in terms of mechanical quantities was due to \cite{korteweg1901forme}. The main assumption was that changes in mass density affect the stresses. Since then, various classes of fluids with such constitutive relations for stresses were named after him. 

Korteweg fluids are not only an example of fluids that inherit capillary effects. They are also a prototype of the fluids in which the boundary between the fluid phase and the solid/gaseous phase is described with a diffuse interface, rather than the sharp one. In this description, the role of the so-called order parameter is played by the gradient of mass density (see \cite{anderson1998diffuse}). 

One of the important questions about Korteweg fluids is thermodynamic consistency of the model, i.e. compatibility of the constitutive relations with the entropy balance law. It turned out that positive answer to this question requires certain deviation from the standard constitutive assumptions about fluids, and modification of the procedures used to establish thermodynamic consistency. Our study is going to make a contribution to the analysis of this problem. 

The foundational work in the analysis of thermodynamic consistency of Korteweg stress is seminal paper of \cite{dunn1985interstitial}. In this study, Korteweg stresses are derived from an interstitial working hypothesis within the Coleman--Noll framework (\cite{coleman1974thermodynamics}). 

Review paper of \cite{anderson1998diffuse} provides a nice overview of the diffuse interface models of fluids. It shows that Korteweg stresses may be determined in two ways. The first one is an equilibrium approach, based upon free energy functional and variational derivatives (see also the paper of \cite{lowengrub1998quasi}). As a consequence, the Korteweg stresses are recovered as equilibrium stresses. Thermodynamic consistency of this approach requires careful justification, as shown by \cite{heida2012cahn}. The second approach essentially relies on the Coleman-Noll procedure, exploited in the way typical for thermodynamics of irreversible processes (TIP). However, to recover the Korteweg stresses it was necessary to assume that specific internal energy depends on the mass density gradient (apart from mass density and temperature). 

The study of \cite{heida2010compressible} analyzes thermodynamic consistency of the class of Korteweg fluid-like materials. Its methodology is based upon the study of \cite{rajagopal2004thermomechanical}, but also inherits certain important new features: (i) constitutive relation for specific internal energy comprises state variables other than density and specific entropy; (ii) entropy flux is assumed in generalized form, not proportional to the heat flux; (iii) constitutive relations are determined using the argument of maximization of the entropy production. Assumption (i) opens the door to capillary effects through introduction of the mass density gradient in a constitutive relation. Consequently,  the balance law for the mass density gradient has to be added to the list of balance laws. It is derived from the gradient of mass balance law. Thanks to this, stress tensor inherited the Korteweg fluid-like terms that model the capillary effects in rather general form. Moreover, assumption (ii) led to two different models for this kind of materials. Certain concepts from this study influenced our approach to the problem. 

In several papers Cimmelli et al. analyzed the models Korteweg fluids (\cite{cimmelli2011thermodynamics}), and their mixtures (\cite{cimmelli2020weakly}). Using similar methodology, \cite{gorgone2021characterization} analyzed their higher-order versions, while \cite{gorgone2021thermodynamical} studied mixtures of viscous Korteweg fluids. One of the aims of these studies was to compare the Coleman-Noll procedure with Liu's method of multipliers. To that end, they generalized the analysis by taking into account gradients of classical balance laws (not only the mass balance law). As a consequence, they formally generalized Liu's procedure and applied it to such system. The outcome of the comparative analysis was that two procedures lead to equivalent set of equations for constitutive functions. Furthermore, they assumed the stress tensor in the form that inherits capillary effects and determined the material coefficients from the constraints imposed by the Coleman-Noll/Liu's procedure. In our opinion, these results do not exploit neither Coleman-Noll's nor Liu's method in their full capacity. 

It has to be emphasized that recent kinetic-theory results of \cite{bhattacharjee2024temperature} shed new light on the structure of Korteweg stresses. In particular, they derived a temperature-dependent Korteweg coefficient from the Enskog--Vlasov equation. This calls for a macroscopic framework admitting such dependence. 

Recent study of \cite{van2023holographic} put for the first time application of the Liu's procedure to Korteweg fluids in a proper perspective, using most of its features in appropriate way. In continuation, the comprehensive Liu-procedure treatment of weakly nonlocal fluids in \cite{szucs2025thermodynamic} provides a framework for comparison, particularly regarding the holographic property of perfect Korteweg fluids. Our work here is compatible with it but emphasizes different aspects and demonstrates the applicability of Liu procedure with a local balance constraints.

Aim of this paper is to analyze the constitutive modelling of a Korteweg fluid by means of Liu's method of multipliers. Although this may seem to be resolved in other previously mentioned studies, it will be shown that certain changes in the procedure lead to improvement of the results and physically appealing outcomes. Here, we provide the list of assumptions which distinguish our approach. 

The specific entropy plays a central role in the analysis. To recover its equilibrium form, and to capture non-equilibrium effects, we shall assume a peculiar functional form. It will depend on the specific internal energy $e$ (as a function of state variables), and also on the state variables independently of $e$. This will allow classical thermodynamic relations to hold, and leave the room for non-equilibrium contributions at the same time. 

The principles of constitutive theory restrict the structure of constitutive functions, and thus simplify the procedure. One of the restrictions is that the symmetric part of the velocity gradient tensor appears in constitutive functions. This motivated the common practice to restrict the analysis to symmetric parts of tensorial quantities from the outset. In this study, the velocity gradient is decomposed into symmetric and skew-symmetric part. Retaining the skew-symmetric part leads to a simpler and more natural derivation of the multipliers related to equations which govern the capillary effects. 

Derivation of the entropy flux, and justification of its generalized form, is always a delicate step in the analysis. It is usually the subject of subtle arguments concerning its relation to the heat flux, either in the form of theorem, or in the form of assumption. Moreover, there is no consensus in the literature about the structure of entropy flux. In our approach, entropy flux comes as a consequence of the residual inequality. In particular, it is determined from the divergence part whose sign cannot be controlled. 

Distinguishing feature of the Korteweg stresses is that they exist in equilibrium state of the system. Consequently, they are regarded as equilibrium stresses, or reversible part of the stress tensor. However, in application of the Coleman-Noll procedure, or Liu's method, they are usually derived from the entropy inequality in general form. Our analysis is based upon the two-step procedure, applied already by Liu (\cite{liu1972method,liu2002continuum}) in the analysis of Newtonian fluids. Namely, the equilibrium part of the stress tensor and heat flux will be computed using the necessary conditions of local minimum of the entropy production in equilibrium. In such a way, we shall retain the equilibrium character of the Korteweg stresses. 

Gibbs' relation has a specific status in thermodynamics. It comes from equilibrium considerations, and it is assumed to be valid in non-equilibrium processes. It is interesting that application of Liu's method does not require this assumption---Gibbs' relation appears as a consequence of the procedure itself. Using this feature, we derived the generalized Gibbs' relation for Korteweg fluids, which takes into account capillary effects. In particular, it contains an extra term corresponding to the infinitesimal change of the density gradient. 

In \cite{anderson1998diffuse} and \cite{heida2010compressible} crucial role in the study of Korteweg fluids is played by the specific internal energy. It is assumed that it inherits the capillary terms, along with the classical thermodynamic part. Moreover, the structure of specific internal energy directly influences the structure of material coefficients in stress tensor. Our approach does correlate specific internal energy and stress tensor directly. However, a comparison will be made which gives an insight into this problem. 

The paper is organized as follows. In Section \ref{Sec:Lemma}, the algebraic Lemma is recalled, which is a cornerstone for further application of the Liu method. Section \ref{Sec:EulerFluids} is concerned with application of the Liu method to Euler fluids. It will serve to introduce, in a simpler context, structural assumptions and procedures that are specific for our approach. Main results of the study are given in Section \ref{Sec:KortewegFluids}, devoted to application of the Liu method to Korteweg fluids. In the final Section \ref{Sec:Conclusions}, a brief summary is provided and an outlook for prospective studies is given. 

%%%%%% 

\section{The algebraic lemma} \label{Sec:Lemma}

At the beginning, we shall recall the algebraic lemma which plays a central role in subsequent analysis. Essentially, it introduces the method of multipliers in the solution of a system of linear algebraic equations subject to a linear inequality. Suppose that we have a linear system of $m$ equations with $n$ unknowns
\begin{equation}
	\label{sistem jednacina}
	A_{ab}X_b+B_a=0,
\end{equation}
and a linear inequality
\begin{equation} \label{nejednakost}
	\alpha_bX_b+\beta\geq 0,
\end{equation}
where $a=1,\ldots,m$, $b=1,\ldots,n$ and the summation convention is used. As indicated, $\mathbf{A}$ is an $m\times n$ matrix, $\mathbf{B}\in\mathbb{R}^m$, $\boldsymbol\alpha\in\mathbb{R}^n$ are vectors, and $\beta \in \mathbb{R}$ is scalar. By $\mathbf{X} \in \mathbb{R}^{n}$ we denote the vector of unknowns. 

\begin{lemma} \label{lemma:Liu-Lemma}
	Let 
	\begin{equation*}
		S=\{\mathbf{X}\in\mathbb{R}^n|\,A_{ab}X_b+B_a=0\},
	\end{equation*}
	the solution set of \eqref{sistem jednacina}, be neither empty nor the whole space $\mathbb{R}^n$ and $\boldsymbol{\alpha}$ is not a zero vector (otherwise the cases are trivial). Then, the following three statements are equivalent: 
	\begin{enumerate}
		\item[1)] Every solution $\mathbf{X}$ of the system \eqref{sistem jednacina} satisfies the inequality \eqref{nejednakost} 
		\begin{equation} \label{mnozioci uslov 1}
			\alpha_bX_b+\beta\geq 0,\quad \forall\mathbf{X}\in S;
		\end{equation}
		\item[2)]
		There exists a vector $\boldsymbol\Lambda\in\mathbb{R}^m$, $\boldsymbol{\Lambda} \neq \mathbf{0}$, such that 
		\begin{equation} \label{mnozioci uslov 2}
			\alpha_bX_b+\beta-\Lambda_a(A_{ab}X_b+B_a)\geq 0 \quad 	
			\forall\mathbf{X}\in\mathbb{R}^n.
		\end{equation}
		\item[3)]
		There exists a vector $\boldsymbol\Lambda\in\mathbb{R}^m$ such that
		\begin{align} 
			\alpha_b-\Lambda_aA_{ab}&=0,
			\label{mnozioci uslov 3-1} \\
			-\Lambda_aB_a+\beta&\geq 0. 
			\label{mnozioci uslov 3-2}
		\end{align}
	\end{enumerate}
\end{lemma}

Instead of giving the proof of Lemma \ref{lemma:Liu-Lemma}, which may be found in \cite{liu1972method,liu2002continuum}, we shall emphasize its main idea. In a usual approach to the problem, inequality \eqref{nejednakost} is treated as a constraint to the solutions of the system \eqref{sistem jednacina}. Here, the inequality \eqref{nejednakost} is regarded as the main equation, while equations \eqref{sistem jednacina} are treated as constraints. This change of perspective is reflected in the choice of multipliers $\Lambda_a$---they are introduced for the equations \eqref{sistem jednacina}, rather than inequality \eqref{nejednakost}. Thus, the problem with constraints is replaced with the one without them, such that \eqref{mnozioci uslov 2} holds. Furthermore, since $\mathbf{X} \in \mathbb{R}^{n}$ are now unconstrained, equation \eqref{mnozioci uslov 2} can be split into the set of equations \eqref{mnozioci uslov 3-1} used to determine the multipliers $\boldsymbol{\Lambda} \in \mathbb{R}^{m}$, and an inequality \eqref{mnozioci uslov 3-2}. This inequality, in conjunction with equation \eqref{mnozioci uslov 3-1}, imposes a constraint on $B_{a}$ and $\beta$. Namely, if every solution of \eqref{sistem jednacina} satisfies \eqref{mnozioci uslov 1}, then $B_{a}$ and $\beta$ cannot be arbitrary. This kind of restriction will have substantial consequences on the analysis of constitutive relations in the subsequent Sections. 

In the context of thermodynamic consistency analysis, algebraic equations \eqref{sistem jednacina} correspond to the balance laws, whereas inequality \eqref{nejednakost} corresponds to the entropy balance law. Vector $\mathbf{X}$ will play the role of process direction vector, while $\Lambda_a$, $B_a$ and $\beta$ will not be constants, but functions of the elements of the constitutive state space. The terminology of the constitutive state space and process direction vector used throughout this paper follows \cite{muschik1996amendment}; see also \cite{hauser2002historical} for a historical account of Liu's procedure.

%%%%%% 

\section{Euler fluids} \label{Sec:EulerFluids}

%%%%%% 

In this Section, we shall firstly present the application of Liu's method to Euler fluids. Although it is well-known in the literature, our aim is to point out certain assumptions and steps in the analysis that are peculiar for our approach. 

\subsection{Conservation laws}

In the analysis of Euler fluids we shall assume that the conservation laws of mass, momentum and energy hold:
\begin{align}
	\label{zakon održanja mase Ojlerovi fluidi}
	& \dfrac{\partial\rho}{\partial t}+\dfrac{\partial}{\partial x_j}\left(\rho v_j\right)=0,
	\\
	\label{zakon održanja količine kretanja Ojlerovi fluidi}
	& \dfrac{\partial}{\partial t}(\rho v_i)+\dfrac{\partial}{\partial x_j}\left(\rho v_i v_j -t_{ij}\right)=0,
	\\
	\label{zakon održanja energije Ojlerovi fluidi}
	& \frac{\partial}{\partial t}\left(\frac{1}{2}\rho\vert\mathbf{v}\vert ^2+\rho e\right)+\dfrac{\partial}{\partial x_j}\left(\left(\frac{1}{2}\rho\vert\mathbf{v}\vert ^2+\rho e\right)v_j-t_{ij}v_i+q_j\right)=0,
\end{align}
where $\rho$ is the mass density, $\mathbf{v}$ is  the velocity, $e$ is the specific internal energy, $\mathbf{t} = \left[ t_{ij} \right]_{ij=1}^{3}$ is the stress tensor, $\mathbf{q} = \left[ q_{i} \right]_{i=1}^{3}$ is the internal energy (heat) flux. Throughout the paper, summation convention will be assumed to hold. 

\subsection{Constitutive assumptions} 

Closure of the system of conservation laws \eqref{zakon održanja mase Ojlerovi fluidi}-\eqref{zakon održanja energije Ojlerovi fluidi} requires constitutive relations for the specific internal energy $e$, and for the non-convective fluxes $t_{ij}$ and $q_{i}$. On one hand, they will be restricted by our choice of the \emph{constitutive state space}---field variables on which constitutive functions depend. To model Euler fluids, it is sufficient to take $\{ \rho, \theta \}$ as the constitutive state space, where $\theta$ is the temperature. This implies the following general form of the constitutive functions: 
\begin{align} \label{KonstitutivneRelacije-Ojler}
	e & = e(\rho, \theta), 
	\nonumber \\
	t_{ij} & = t_{ij}(\rho, \theta), 
	\\ 
	q_{i} & = q_{i}(\rho, \theta). 
	\nonumber
\end{align}
Since $t_{ij}$ and $q_{i}$ are the components of a Galilean tensor and vector, respectively, and they depend only on scalars $\rho$ and $\theta$, they have to be of the following form:
\begin{equation} \label{KonstitutivneRelacije-Ojler2}
	t_{ij} = - p(\rho, \theta) \delta_{ij}, \quad q_{i} = 0,  
\end{equation}
where $p$ is the thermodynamic pressure. This is in agreement with the assumption that Euler fluids are inviscid and do not transfer heat. 

Constitutive relations have to be chosen in such a way that conservation laws \eqref{zakon održanja mase Ojlerovi fluidi}--\eqref{zakon održanja energije Ojlerovi fluidi} are compatible with the entropy balance law: 
\begin{equation} \label{bilans entropije Ojlerovi fluidi}
	\dfrac{\partial}{\partial t}\left(\rho s\right)+\dfrac{\partial}{\partial x_j}\left(\rho sv_j+\varphi_j\right)=\Sigma (\geq 0). 
\end{equation}
In contrast to Coleman-Noll-like procedures, in which the structure of the specific entropy $s$ and the entropy flux $\varphi_{j}$ are partially or completely prescribed, the method of multipliers does not impose such a restriction beforehand. Instead, they are determined by the constitutive functions, as well: 
\begin{equation} \label{konstitutivne relacije s-phi Ojler}
	s = s(\rho, \theta), \quad \varphi_{j} = \varphi_{j}(\rho, \theta). 
\end{equation}
Besides this assumption, an additional one will be introduced regarding the form of the constitutive function for specific entropy $s$: 
\begin{equation} \label{konstitutivne relacije s hipoteza 1 Ojlerovi fluidi}
	s = s(e(\rho, \theta), \rho). 
\end{equation}
It is motivated by the aim to recover the classical equilibrium thermodynamic relations: 
\begin{equation} \label{konstitutivne relacije s hipoteza 2 Ojlerovi fluidi}
	\dfrac{\partial s}{\partial e} = \dfrac{1}{\theta}, \quad 
	\dfrac{\partial s}{\partial \rho} = - \dfrac{1}{\theta} \dfrac{p}{\rho^{2}}, 
\end{equation}
which we shall assume to hold. Note that \eqref{konstitutivne relacije s hipoteza 1 Ojlerovi fluidi} implies only implicit dependence on temperature $\theta$, whereas $s$ may depend on $\rho$ both implicitly and explicitly. To that end, the following notation will be introduced for the total derivative of specific entropy with respect to constitutive variables:
\begin{equation} \label{totalni izvod Ojlerovi fluidi}
	\overline{\dfrac{\partial s}{\partial \rho}} 
	= \dfrac{\partial s}{\partial e}\dfrac{\partial e}{\partial \rho} 
	+ \dfrac{\partial s}{\partial \rho}, 
	\quad
	\overline{\dfrac{\partial s}{\partial \theta}} 
	= \dfrac{\partial s}{\partial e}\dfrac{\partial e}{\partial \theta}. 
\end{equation} 

\subsection{The method of multipliers} 

The method of multipliers is a specific, formal procedure which makes the balance laws compatible with the entropy balance law. This specific formulation of the entropy principle was proposed by Liu. According to this approach, the entropy balance law \eqref{bilans entropije Ojlerovi fluidi} is the main equation which determines the evolution of thermodynamic process, whereas the balance laws \eqref{zakon održanja mase Ojlerovi fluidi}-\eqref{zakon održanja energije Ojlerovi fluidi} are the constraints that have to be satisfied. To transform the problem with constraints into the one without them, Lagrange multipliers $\Lambda^\rho$, $\Lambda^{v_i}$ and $\Lambda^e$ are introduced and the entropy balance law is written in extended form
\begin{align} \label{prosirena entropijska nejednakost Ojlerovi fluidi}
	\dfrac{\partial}{\partial t}\left(\rho s\right)
	& +\dfrac{\partial}{\partial x_j}\left(\rho sv_j+\varphi_j\right)
	-\Lambda^\rho\left[\dfrac{\partial\rho}{\partial t}+\dfrac{\partial}{\partial x_j}\left(\rho v_j\right)\right]
	\nonumber \\
	&-\Lambda^{v_i}\left[\dfrac{\partial}{\partial t}(\rho v_i)+\dfrac{\partial}{\partial x_j}\left(\rho v_i v_j -t_{ij}\right)\right] 
	\\
	&-\Lambda^ e\left[\frac{\partial}{\partial t}\left(\frac{1}{2}\rho\vert\mathbf{v}\vert ^2+\rho e\right)+\dfrac{\partial}{\partial x_j}\left(\left(\frac{1}{2}\rho\vert\mathbf{v}\vert ^2+\rho e\right)v_j-t_{ij}v_i+q_j\right)\right] 
	\nonumber \\
	& =\Sigma\geq 0. 
	\nonumber
\end{align}
Taking into account constitutive relations \eqref{KonstitutivneRelacije-Ojler} and \eqref{konstitutivne relacije s-phi Ojler}, and assumption \eqref{konstitutivne relacije s hipoteza 1 Ojlerovi fluidi}, the following extended entropy balance law is obtained: 
\begin{align} \label{produkcija entropije Ojlerovi fluidi}
	& \dfrac{\partial\rho}{\partial t}s+\rho\dfrac{\partial s}{\partial t}+\dfrac{\partial\rho}{\partial x_j}sv_j+\rho\dfrac{\partial s}{\partial x_j}v_j+\rho s \dfrac{\partial v_j}{\partial x_j}
	+\dfrac{\partial\varphi_j}{\partial x_j}-\Lambda^\rho\left[\dfrac{\partial\rho}{\partial t}+\dfrac{\partial\rho}{\partial x_j}v_j+\rho\dfrac{\partial v_j}{\partial x_j}\right]
	\nonumber \\
	& \quad -\Lambda^{v_i}\left[\dfrac{\partial\rho}{\partial t}v_i+\rho\dfrac{\partial v_i}{\partial t}+\dfrac{\partial\rho}{\partial x_j}v_iv_j+\rho\dfrac{\partial v_i}{\partial x_j}v_j+\rho v_i\dfrac{\partial v_j}{\partial x_j}-\dfrac{\partial t_{ij}}{\partial x_j}\right]
	\\
	& \quad -\Lambda^ e\left[\dfrac{1}{2}\dfrac{\partial\rho}{\partial t}|\mathbf{v}|^2+\rho\dfrac{\partial v_i}{\partial t}v_i+\dfrac{\partial\rho}{\partial t} e+\rho\left(\dfrac{\partial e}{\partial\rho}\dfrac{\partial\rho}{\partial t}+\dfrac{\partial e}{\partial\theta}\dfrac{\partial\theta}{\partial t}\right)\right.
	\nonumber \\
	& \qquad \left.+\left(\dfrac{1}{2}\dfrac{\partial\rho}{\partial x_j}|\mathbf{v}|^2+\rho\dfrac{\partial v_i}{\partial x_j}v_i+\dfrac{\partial\rho}{\partial x_j} e+\rho\left(\dfrac{\partial e}{\partial\rho}\dfrac{\partial\rho}{\partial x_j}+\dfrac{\partial e}{\partial\theta}\dfrac{\partial\theta}{\partial x_j}\right)\right)v_j\right. 
	\nonumber \\
	& \qquad \left.+\left(\dfrac{1}{2}\rho|\mathbf{v}|^2+\rho
	e\right)\dfrac{\partial v_j}{\partial x_j}-\dfrac{\partial t_{ij}}{\partial x_j}v_i-t_{ij}\dfrac{\partial v_i}{\partial x_j}\right]=\Sigma\geq 0 
	\nonumber
\end{align}
Using the identities \eqref{App:Lij}, we can transform the extended entropy inequality into the form: 
\begin{align} \label{ExtEnropyInequality-Euler}
	& \dfrac{\partial \rho}{\partial t} \left[ s + \rho \overline{\dfrac{\partial s}{\partial\rho}} - \Lambda^\rho - \Lambda^{v_i}v_i - \Lambda^ e \left(\dfrac{1}{2}|\mathbf{v}|^2 + e + \rho\dfrac{\partial e}{\partial\rho}\right) \right] 
	\nonumber \\
	& \quad + \dfrac{\partial v_i}{\partial t} \left[- \Lambda^{v_i}\rho - \Lambda^ e \rho v_i \right]
	\nonumber \\ 
	& \quad + \dfrac{\partial \theta}{\partial t} \left[ \rho \overline{\dfrac{\partial s}{\partial\theta}} - \Lambda^e \rho \dfrac{\partial e}{\partial\theta} \right] 
	\nonumber \\ 
	& \quad + \dfrac{\partial \rho}{\partial x_j} \left[ sv_j + \rho \overline{\dfrac{\partial s}{\partial\rho}} v_j + \dfrac{\partial\varphi_j}{\partial\rho} - \Lambda^\rho v_j - \Lambda^{v_i}v_i v_j - \Lambda^{v_j} \dfrac{\partial p}{\partial\rho} %\right.
	%\nonumber\\ & \quad \left.
	- \Lambda^e \left(\left(\dfrac{1}{2}|\mathbf{v}|^2 + e + \rho\dfrac{\partial e}{\partial\rho} \right) v_j + \dfrac{\partial p}{\partial\rho} v_j \right) \right]
	\\ 
	& \quad + \dfrac{\partial \theta}{\partial x_j} \left[ \rho \overline{\dfrac{\partial s}{\partial\theta}} v_j + \dfrac{\partial\varphi_j}{\partial\theta} - \Lambda^{v_j}\dfrac{\partial p}{\partial\theta} - \Lambda^e \left(\rho\dfrac{\partial e}{\partial\theta}v_j + \dfrac{\partial p}{\partial\theta} v_j \right) \right] 
	\nonumber \\ 
	& \quad + D_{ij} \Big[ \rho s\delta_{ij} - \Lambda^\rho \rho \delta_{ij} - \Lambda^{v_i} \rho v_j - \Lambda^{v_k} \rho v_k \delta_{ij} %\right. 
	% \nonumber \\ & \quad 
	- \Lambda^e \rho v_i v_j - \Lambda^e \left( \dfrac{1}{2} \rho|\mathbf{v}|^2 + \rho e \right) \delta_{ij} - \Lambda^e p \delta_{ij} \Big] 
	\nonumber \\ 
	& \quad + W_{ij} \left[ -\Lambda^{v_i}\rho v_j-\Lambda^{e}\rho v_i v_j \right] = \Sigma \geq 0. 
	\nonumber
\end{align}
In this form we can recognize the process direction vector: 
\begin{equation} \label{ProcessDirectionVector-Euler}
	\mathbf{X} = \left( \dfrac{\partial \rho}{\partial t}, \dfrac{\partial v_i}{\partial t}, 
	\dfrac{\partial \theta}{\partial t}, \dfrac{\partial \rho}{\partial x_j}, 
	\dfrac{\partial v_i}{\partial x_j}, \dfrac{\partial \theta}{\partial x_j} \right), 
\end{equation}
which plays the role of unknown (field) variables, equivalent to the vector $\mathbf{X}$ in Liu's Lemma \ref{lemma:Liu-Lemma}. Note that velocity gradient $\dfrac{\partial v_{i}}{\partial x_{j}} = L_{ij}$ is in \eqref{ExtEnropyInequality-Euler} decomposed into symmetric part $D_{ij}$ and skew-symmetric part $W_{ij}$ according to \eqref{App:Lij}. 

Components of the process direction vector \eqref{ProcessDirectionVector-Euler} are linearly independent and do not belong to the constitutive state space. Therefore, according to Lemma \ref{lemma:Liu-Lemma}, their coefficients must vanish. This will lead to the set of equations which will serve to compute the multipliers and (eventually) the entropy flux, and the \emph{residual inequality} which is the basis for derivation of the constitutive relations. 

First group of equations comes from the coefficients of time derivatives $\dfrac{\partial\rho}{\partial t}$, $\dfrac{\partial v_i}{\partial t}$ and $\dfrac{\partial\theta}{\partial t}$:
\begin{align}
	\label{Trho Ojlerovi fluidi}
	& s+\rho\overline{\dfrac{\partial s}{\partial\rho}}-\Lambda^\rho-\Lambda^{v_i}v_i-\Lambda^ e\left(\dfrac{1}{2}|\mathbf{v}|^2+e+\rho\dfrac{\partial e}{\partial\rho}\right)=0, 
	\\
	\label{Tvi Ojlerovi fluidi}
	& -\Lambda^{v_i}\rho-\Lambda^ e\rho v_i=0, \quad i=1,2,3, 
	\\
	\label{Ttheta Ojlerovi fluidi}
	& \rho\overline{\dfrac{\partial s}{\partial\theta}}-\Lambda^ e\rho\dfrac{\partial e}{\partial\theta}=0.
\end{align}
It will serve to determine the multipliers $\Lambda^\rho$, $\Lambda^{v_i}$ and $\Lambda^e$. 
Second group of equations comes from the coefficients of space derivatives $\dfrac{\partial\rho}{\partial x_j}$, $\dfrac{\partial\theta}{\partial x_j}$ and $D_{ij}$ $(i,j = 1, 2, 3)$:
\begin{align} 
	\label{Xrho Ojlerovi fluidi}
	& sv_j +\rho\overline{\dfrac{\partial s}{\partial\rho}} v_j+\dfrac{\partial\varphi_j}{\partial\rho}-\Lambda^\rho v_j-\Lambda^{v_i}v_iv_j-\Lambda^{v_j}\dfrac{\partial p}{\partial\rho} 
	% \\ & \quad 
	-\Lambda^ e\left(\left(\dfrac{1}{2}|\mathbf{v}|^2+e+\rho\dfrac{\partial e}{\partial\rho}\right)v_j+\dfrac{\partial p}{\partial\rho}v_j\right)=0, %\quad j=1,2,3, 
	\\
	\label{Xtheta Ojlerovi fluidi}
	& \rho\overline{\dfrac{\partial s}{\partial\theta}} v_j+\dfrac{\partial\varphi_j}{\partial\theta}-\Lambda^{v_j}\dfrac{\partial p}{\partial\theta}-\Lambda^ e\left(\rho\dfrac{\partial e}{\partial\theta}v_j+\dfrac{\partial p}{\partial\theta}v_j\right)=0, % \quad j=1,2,3. 
	\\
	\label{Dij Ojlerovi fluidi}
	& \rho s\delta_{ij} - \Lambda^\rho \rho \delta_{ij} - \Lambda^{v_i} \rho v_j - \Lambda^{v_k} \rho v_k \delta_{ij} %\right. 
	% \\ & \quad 
	- \Lambda^e \rho v_i v_j - \Lambda^e \left( \dfrac{1}{2} \rho|\mathbf{v}|^2 + \rho e \right) \delta_{ij} - \Lambda^e p \delta_{ij} = 0. 
	%\nonumber 
\end{align}
It will help to determine the entropy flux $\varphi_{j}$ and (possibly) impose restrictions on multipliers. The remaining part of equation \eqref{ExtEnropyInequality-Euler} contains $W_{ij}$, which is skew-symmetric and neither belongs to the constitutive state space $\{ \rho, \theta \}$, nor to the process direction vector \eqref{ProcessDirectionVector-Euler}. Therefore, we shall impose the following condition: 
\begin{equation} \label{Wij Ojlerovi fluidi}
	\left[ -\Lambda^{v_i}\rho v_j-\Lambda^{e}\rho v_i v_j \right] W_{ij} = 0. 
\end{equation}

It is clear that application of Lemma \ref{lemma:Liu-Lemma}, which imposed equations \eqref{Trho Ojlerovi fluidi}-\eqref{Wij Ojlerovi fluidi}, leads to a trivial residual inequality, $\Sigma = 0$. This appears to be a consequence of the choice of constitutive state space. This will be expressed in the form of Lemma in the sequel. 

\subsection{Computation of the multipliers} 

By solving equations \eqref{Trho Ojlerovi fluidi}--\eqref{Ttheta Ojlerovi fluidi} we shall determine the multipliers. This will be done in two steps. 
\begin{lemma} \label{Lema:Lambda1-Ojler}
	The multipliers $\Lambda^{\rho}$ and $\Lambda^{v_i}$ have the following form: 
	\begin{align}
		\label{modifikovana Trho Ojlerovi fluidi}
		\Lambda^\rho & = s + \rho\overline{\dfrac{\partial s}{\partial\rho}} 
		+ \dfrac{1}{2}\Lambda^e|\mathbf{v}|^2 - \Lambda^ e e 
		- \Lambda^e \rho \dfrac{\partial e}{\partial\rho}, 
		\\
		\label{Tv_i' Ojlerovi fluidi}
		\Lambda^{v_i} & = -\Lambda^e v_i, \quad i = 1, 2, 3.
	\end{align}
\end{lemma}

\begin{proof}
	Since $\rho > 0$, equation \eqref{Tv_i' Ojlerovi fluidi} follows directly from \eqref{Tvi Ojlerovi fluidi}. Inserting \eqref{Tv_i' Ojlerovi fluidi} into \eqref{Trho Ojlerovi fluidi}, one obtains \eqref{modifikovana Trho Ojlerovi fluidi}. 
\end{proof}

Lemma \ref{Lema:Lambda1-Ojler} expresses the multipliers $\Lambda^\rho$ and $\Lambda^{v_i}$ in terms of $\Lambda^e$. In the sequel, we shall determine their final form. 
\begin{theorem} \label{Te:Lambda-e Ojler}
	Multiplier $\Lambda^{e}$ of equation \eqref{zakon održanja energije Ojlerovi fluidi} reads: 
	\begin{equation} \label{mnozitelj Lambda-e Ojlerovi fluidi}
		\Lambda^{e} = \dfrac{1}{\theta}. 
	\end{equation}
	%Besides, the following relation holds: 
	%\begin{equation} \label{s po theta Ojlerovi fluidi}
	%\dfrac{\partial s}{\partial \theta} = 0.
	%\end{equation}
\end{theorem}

\begin{proof}
	Using the total derivative of specific entropy \eqref{totalni izvod Ojlerovi fluidi}$_{2}$, equation \eqref{Ttheta Ojlerovi fluidi} becomes: 
	\begin{equation}
		\label{Ttheta pomoćna Ojlerovi fluidi}
		\rho\left(\dfrac{\partial s}{\partial e} 
		- \Lambda^e\right) \dfrac{\partial e}{\partial\theta} = 0.
		%+ \rho\dfrac{\partial s}{\partial\theta} = 0.
	\end{equation}
	This implies:
	\begin{equation} \label{Lambda_e Ojlerovi fluidi}
		\Lambda^e = \dfrac{\partial s}{\partial e} \overset{\eqref{konstitutivne relacije s hipoteza 2 Ojlerovi fluidi}_{1}}{=} \dfrac{1}{\theta}, 
	\end{equation} 
	which completes the proof.
\end{proof} 

The technique used for solving equation \eqref{Ttheta Ojlerovi fluidi} and determination of $\Lambda^e$ is facilitated by assumption \eqref{konstitutivne relacije s hipoteza 1 Ojlerovi fluidi} about the structure of specific entropy. This will become essential in the study of Korteweg fluids. 

\begin{corollary} \label{Cor:Mnozitelji Ojlerovi fluidi}
	Multipliers $\Lambda^{\rho}$ and $\Lambda^{v_i}$ of equations \eqref{zakon održanja mase Ojlerovi fluidi} and \eqref{zakon održanja količine kretanja Ojlerovi fluidi}, respectively, are determined by the following relations: 
	\begin{align} 
		\label{mnozitelj Lambda-rho Ojlerovi fluidi}
		\Lambda^{\rho} & = - \dfrac{1}{\theta} \left( e - \theta s + \frac{p}{\rho} 
		- \frac{1}{2} |\mathbf{v}|^{2} \right), 
		\\
		\label{mnozitelj Lambda-vi Ojlerovi fluidi}
		\Lambda^{v_i} & = - \dfrac{v_{i}}{\theta}, \quad i = 1, 2, 3.
	\end{align}
\end{corollary}
\begin{proof}
	Using the total derivative \eqref{totalni izvod Ojlerovi fluidi} of specific entropy, 
	\begin{align*}
		\Lambda^\rho &\overset{\eqref{modifikovana Trho Ojlerovi fluidi}}{=} s + \rho\overline{\dfrac{\partial s}{\partial\rho}} 
		+ \dfrac{1}{2}\Lambda^e|\mathbf{v}|^2 - \Lambda^ e e 
		- \Lambda^e \rho \dfrac{\partial e}{\partial\rho} 
		\\
		&\overset{\eqref{Lambda_e Ojlerovi fluidi}}{=} s + \rho\left(\dfrac{\partial s}{\partial e}\dfrac{\partial e}{\partial\rho}+\dfrac{\partial s}{\partial\rho}\right)+\dfrac{1}{2}\dfrac{\partial s}{\partial e}|\mathbf{v}|^2 -\dfrac{\partial s}{\partial e} e - \dfrac{\partial s}{\partial e} \rho \dfrac{\partial e}{\partial\rho}.	   
	\end{align*}
	Taking into account \eqref{konstitutivne relacije s hipoteza 2 Ojlerovi fluidi}, equation \eqref{mnozitelj Lambda-rho Ojlerovi fluidi} follows. On the other hand, \eqref{mnozitelj Lambda-vi Ojlerovi fluidi} is obtained from \eqref{Tv_i' Ojlerovi fluidi} taking into account \eqref{mnozitelj Lambda-e Ojlerovi fluidi}. 
\end{proof} 

\begin{corollary} \label{Cor:Dij Wij Ojlerovi fluidi}
	For multipliers $\Lambda^{\rho}$, $\Lambda^{v_i}$, and $\Lambda^{e}$, determined by \eqref{mnozitelj Lambda-rho Ojlerovi fluidi}, \eqref{mnozitelj Lambda-vi Ojlerovi fluidi}, and \eqref{mnozitelj Lambda-e Ojlerovi fluidi}, respectively, the following equations hold: 
	%\begin{equation} 
	\begin{align} \label{Dij_Ojlerovi fluidi}
		& \rho s\delta_{ij} -\Lambda^\rho\rho\delta_{ij}-\Lambda^{v_i}\rho v_j-\Lambda^{v_k}\rho   
		v_k\delta_{ij} %\\ 
		-\Lambda^e\rho v_iv_j-\Lambda^ e\left(\dfrac{1}{2}\rho|\mathbf{v}|^2 
		+\rho e\right)\delta_{ij}-\Lambda^ e p\delta_{ij} = 0, 
		\\
		% \nonumber 
	%\end{align}
	%\end{equation}
	%\begin{equation} 
		\label{Wij_Ojlerovi fluidi}
		& -\Lambda^{v_i}\rho v_j-\Lambda^{e}\rho v_i v_j = 0. 
	\end{align}
	
\end{corollary}
\begin{proof}
	Equations \eqref{Dij_Ojlerovi fluidi} and \eqref{Wij_Ojlerovi fluidi} are proved by direct substitution of the multipliers \eqref{mnozitelj Lambda-e Ojlerovi fluidi}, \eqref{mnozitelj Lambda-vi Ojlerovi fluidi} and \eqref{mnozitelj Lambda-vi Ojlerovi fluidi}. In such a way, equations \eqref{Dij Ojlerovi fluidi} and \eqref{Wij Ojlerovi fluidi} are reduced to identities. 
\end{proof}

Although decomposition of the velocity gradient in the last result seems redundant, it will be crucial for the analysis of Korteweg fluids. Symmetric part will enter into the entropy production, whereas the skew-symmetric part will facilitate computation of the additional multiplier. 

\subsection{Entropy flux and entropy production} 

Consequences of the extended entropy balance law also determine the structure of the constitutive functions \eqref{konstitutivne relacije s-phi Ojler}. We shall determine the entropy flux $\varphi_{j}$ and entropy production $\Sigma$. 

\begin{lemma} \label{Lema:phi_j jednakosti Ojler}
	The following relations hold: 
	\begin{equation} \label{phi izvodi Ojler}
		\dfrac{\partial\varphi_j}{\partial\rho}=0, \quad 
		\dfrac{\partial\varphi_j}{\partial\theta} = 0, 
	\end{equation}
	which implies: 
	\begin{equation} \label{phi Ojler}
		\varphi_{j} = 0, \quad j = 1, 2, 3. 
	\end{equation}
\end{lemma}
\begin{proof}
	Firstly, it will be proved that entropy flux is $\varphi_j = \mathrm{const.}$ Rewrite equation \eqref{Xrho Ojlerovi fluidi} in the form:
	\begin{equation*}
		\left(s+\rho\left(\dfrac{\partial s}{\partial e}\dfrac{\partial e}{\partial\rho}+\dfrac{\partial s}{\partial\rho}\right)
		- \Lambda^\rho-\Lambda^{v_i}v_i -\Lambda^e \left( \dfrac{1}{2}|\mathbf{v|^2} + e + \rho\dfrac{\partial e}{\partial\rho} \right) \right) v_j %\\
		+\dfrac{\partial\varphi_j}{\partial\rho}
		-\left(\Lambda^{v_j}+\Lambda^ e v_j\right)\dfrac{\partial p}{\partial\rho} = 0. 
	\end{equation*}
	Using \eqref{Trho Ojlerovi fluidi} and \eqref{Tv_i' Ojlerovi fluidi} it is reduced to \eqref{phi izvodi Ojler}$_{1}$. Rewriting equation \eqref{Xtheta Ojlerovi fluidi} in the form: 
	\begin{equation*}
		\rho \left( \dfrac{\partial s}{\partial e} - \Lambda^e\right) 
		\dfrac{\partial e}{\partial\theta} v_j + \dfrac{\partial\varphi_j}{\partial\theta}-(\Lambda^{v_j}+\Lambda^ e v_j)\dfrac{\partial p}{\partial\theta}=0,
	\end{equation*}
	and applying \eqref{Tv_i' Ojlerovi fluidi} and \eqref{Ttheta pomoćna Ojlerovi fluidi}, \eqref{phi izvodi Ojler}$_{2}$ is obtained. 
	
	Taking into account assumption \eqref{konstitutivne relacije s-phi Ojler}$_{2}$, i.e. $\varphi_{j} = \varphi_{j}(\rho, \theta)$, equation \eqref{phi izvodi Ojler} implies $\varphi_{j} = \mathrm{const.}$ However, since entropy flux is relevant for the local entropy balance law \eqref{bilans entropije Ojlerovi fluidi} only through its divergence, we may take \eqref{phi Ojler} without loss of generality. 
\end{proof}

\begin{lemma} \label{Lema:produkcija entropije Ojler}
	Entropy production in Euler fluids vanishes: 
	\begin{equation} \label{Sigma Ojlerovi fluidi}
		\Sigma = 0. 
	\end{equation}
\end{lemma} 

\begin{proof}
	As consequence of Theorem \ref{Te:Lambda-e Ojler}, Corollaries \ref{Cor:Mnozitelji Ojlerovi fluidi} and \ref{Cor:Dij Wij Ojlerovi fluidi}, and Lemma \ref{Lema:phi_j jednakosti Ojler}, the left hand side of equation \eqref{ExtEnropyInequality-Euler} vanishes. This implies vanishing of the entropy production \eqref{Sigma Ojlerovi fluidi}. 
\end{proof} 

The result expressed in Lemma \ref{Lema:produkcija entropije Ojler} is typical for Euler fluids. It comes as a consequence of a restrictive choice of the constitutive state space, consisted of the state variables $\{ \rho, \theta \}$ only. 

\subsection{Gibbs relation for Euler fluids}

The final step of this part of the study is derivation of the Gibbs relation. 

\begin{theorem} 
	Gibbs relation for Euler fluids reads: 
	\begin{equation} \label{Gibsova relacija Ojlerovi fluidi}
		\mathrm{d}s = \dfrac{1}{\theta} \mathrm{d}e
		- \dfrac{1}{\theta}\dfrac{p}{\rho^2} \mathrm{d}\rho.
	\end{equation}
\end{theorem}
\begin{proof}
	Using assumption \eqref{konstitutivne relacije s hipoteza 1 Ojlerovi fluidi}, thermodynamic relations \eqref{konstitutivne relacije s hipoteza 2 Ojlerovi fluidi}, and total derivatives \eqref{totalni izvod Ojlerovi fluidi}, total differential of the specific entropy may be written as:
	\begin{align*}
		\mathrm{d}s & = \overline{\dfrac{\partial s}{\partial \rho}} \mathrm{d}\rho 
		+ \overline{\dfrac{\partial s}{\partial \theta}} \mathrm{d}\theta 
		\\ 
		& \overset{\eqref{totalni izvod Ojlerovi fluidi}}{=} 
		\left( \dfrac{\partial s}{\partial e} \dfrac{\partial e}{\partial \rho} 
		+ \dfrac{\partial s}{\partial \rho} \right) \mathrm{d}\rho 
		+ \dfrac{\partial s}{\partial e} \dfrac{\partial e}{\partial \theta} \mathrm{d}\theta
		\\ 
		& = \dfrac{\partial s}{\partial e} \left( \dfrac{\partial e}{\partial \rho} \mathrm{d}\rho 
		+ \dfrac{\partial e}{\partial \theta} \mathrm{d}\theta \right) 
		+ \dfrac{\partial s}{\partial \rho} \mathrm{d} \rho 
		\\
		& \overset{\eqref{konstitutivne relacije s hipoteza 2 Ojlerovi fluidi}}{=} 
		\dfrac{1}{\theta} \mathrm{d}e
		- \dfrac{1}{\theta}\dfrac{p}{\rho^2} \mathrm{d}\rho
	\end{align*}
	which recovers well-known Gibbs relation \eqref{Gibsova relacija Ojlerovi fluidi}. 
\end{proof}

It is important feature of the method of multipliers that Gibbs relation is obtained as a consequence of compatibility of the balance laws and the entropy inequality. On the other hand, in thermodynamics of irreversible processes it regarded as assumption.

Let us remark that the derivation uses the thermodynamic relations \eqref{konstitutivne relacije s hipoteza 2 Ojlerovi fluidi} as input. Therefore, what is obtained here is the compatibility of the Gibbs relation with the entropy inequality and the balance laws, in the same spirit as in Liu's original work (\cite{liu1972method,liu2002continuum}) and \cite{muller1985thermodynamics}. The genuine novelty will appear in Section~\ref{Sec:GibbsKorteweg}, where the \emph{generalized} Gibbs relation including capillary terms is derived---this extension is a true consequence of the procedure.

%%%%%%

\section{Korteweg fluids} \label{Sec:KortewegFluids}

%%%%%% 

In this Section we focus on our main goal---application of the Liu method to Korteweg fluids. This analysis will confirm thermodynamic consistency of Korteweg fluids, i.e. their compatibility with entropy inequality, but it will also bring new results. Main features that distinguish our results from previous studies are the following: 
\begin{itemize}
	\item[(a)] Assumptions and methodology applied here, especially decomposition of the velocity gradient, significantly simplify the procedure of computation of the multipliers. 
	\item[(b)] Entropy flux is determined from the residual inequality, i.e. by annihilation the divergence term whose sign cannot be controlled. 
	\item[(c)] Korteweg stresses are derived from the equilibrium conditions, thus recovering their reversible character. 
	\item[(d)] In view of our choice of the constitutive state space, the nonequilibrium thermodynamic fluxes will be determined in the most general form. 
	\item[(e)] The generalized Gibbs relation is derived, which takes into account the contribution of capillary effects.
\end{itemize}
Finally, the structure of specific internal energy will be discussed and compared with the results of other studies. 

\subsection{Balance laws}

In the modelling of Korteweg fluids we shall assume that conservation laws of mass, momentum and energy hold
\begin{align}
	\label{zakon održanja mase Kortevegovi fluidi}
	& \dfrac{\partial\rho}{\partial t}+\dfrac{\partial}{\partial x_j}\left(\rho v_j\right)=0, 
	\\
	\label{zakon održanja količine kretanja Kortevegovi fluidi}
	& \dfrac{\partial}{\partial t}(\rho v_i)+\dfrac{\partial}{\partial x_j}\left(\rho v_i v_j -t_{ij}\right)=0, 
	\\ 
	\label{zakon održanja energije Kortevegovi fluidi}
	& \frac{\partial}{\partial t}\left(\frac{1}{2}\rho\vert\mathbf{v}\vert ^2+\rho e\right)+\dfrac{\partial}{\partial x_j}\left(\left(\frac{1}{2}\rho\vert\mathbf{v}\vert ^2+\rho e\right)v_j-t_{ij}v_i+q_j\right)=0,
\end{align}
using the same notation as in equations \eqref{zakon održanja mase Ojlerovi fluidi}-\eqref{zakon održanja energije Ojlerovi fluidi} for Euler fluids. Capillary effects will be modelled through the mass density gradient as a field variable (order parameter). To that end its evolution equation/balance law has to be added to the classical ones. It is obtained as gradient of the mass conservation law \eqref{zakon održanja mase Kortevegovi fluidi}, given in three equivalent forms 
\begin{align}
	& \dfrac{\partial^2}{\partial x_i\partial t} \rho + \dfrac{\partial^2}{\partial x_i\partial x_j}(\rho v_j) = 0,
	\nonumber \\
	& \dfrac{\partial}{\partial t}\dfrac{\partial\rho}{\partial x_i} + \dfrac{\partial}{\partial x_j}\dfrac{\partial}{\partial x_i}(\rho v_j) = 0, 
	\label{grad zom} \\
	& \dfrac{\partial}{\partial t}\dfrac{\partial\rho}{\partial x_i} 
	+ \dfrac{\partial}{\partial x_j}\left(\dfrac{\partial\rho}{\partial x_i}v_j + \rho\dfrac{\partial v_j}{\partial x_i}\right) = 0. 
	\nonumber
\end{align}
Using decomposition of the velocity gradient \eqref{App:Lij}, equation $\eqref{grad zom}_3$ becomes
\begin{equation} \label{gradijent odrzanja mase}
	\dfrac{\partial}{\partial t}\rho_{,i}+\dfrac{\partial}{\partial x_j}(\rho_{,i}v_j+\rho (D_{ji}+W_{ji}))=0, \quad i=1,2,3.
\end{equation}
In the sequel we shall study the system of equations \eqref{zakon održanja mase Kortevegovi fluidi}--\eqref{zakon održanja energije Kortevegovi fluidi}, \eqref{gradijent odrzanja mase}. 

\subsection{Constitutive assumptions} 

Closure of the aforementioned system of governing equations requires constitutive relations for the specific internal energy $e$, stress tensor $t_{ij}$ and heat flux $q_{j}$. Constitutive functions have to depend on the quantities (field variables) which determine the \emph{constitutive state space}. In the case of Korteweg fluids, which inherit the information on capillary stresses, i.e. mass density gradient, constitutive state space will be 
\begin{equation} \label{CSS_Korteweg}
	\mathcal{C} = \{ \rho, \rho_{,k}, \rho_{,kl}, \theta, \theta_{,k}, D_{kl} \},  
\end{equation}
where comma denotes the partial derivative with respect to Cartesian coordinate(s). 
Therefore, general form of the constitutive relations will be
\begin{align} \label{konstitutivne relacije 1 opste}
	e & = e(\rho, \rho_{,k},  \rho_{,kl}, \theta, \theta_{,k}, D_{kl}), 
	\nonumber \\
	t_{ij} & = t_{ij}(\rho, \rho_{,k},  \rho_{,kl}, \theta, \theta_{,k}, D_{kl}), 
	\\ 
	q_{j} & = q_{j}(\rho, \rho_{,k},  \rho_{,kl}, \theta, \theta_{,k}, D_{kl}). 
	\nonumber
\end{align}

In accordance with the guiding principle of continuum mechanics, constitutive relations have to be chosen in such a way that balance laws \eqref{zakon održanja mase Kortevegovi fluidi}-\eqref{zakon održanja energije Kortevegovi fluidi} and \eqref{gradijent odrzanja mase} are compatible with the entropy balance law, 
\begin{equation} \label{bilans entropije Kortevegovi fluidi}
	\dfrac{\partial}{\partial t}\left(\rho s\right)+\dfrac{\partial}{\partial x_j}\left(\rho sv_j+\varphi_j\right) = \Sigma, \quad \Sigma \geq 0.
\end{equation}
i.e. they have to ensure that entropy inequality is satisfied for every thermodynamic process determined by \eqref{zakon održanja mase Kortevegovi fluidi}-\eqref{zakon održanja energije Kortevegovi fluidi} and \eqref{gradijent odrzanja mase}. It is peculiar for the method of multipliers that specific entropy $s$ and entropy flux $\varphi_{j}$ are not given in advance. They also have to be determined by means of constitutive relations
\begin{align} \label{konstitutivne relacije s-phi opste}
	s & = s(\rho, \rho_{,k},  \rho_{,kl}, \theta, \theta_{,k}, D_{kl}), \\
	\varphi_j & = \varphi_j(\rho, \rho_{,k},  \rho_{,kl}, \theta, \theta_{,k}, D_{kl}). 
	\nonumber 
\end{align}

Additional assumptions will be introduced regarding the constitutive function of specific entropy. On one hand, they generalize the structure proposed in \eqref{konstitutivne relacije s hipoteza 1 Ojlerovi fluidi}. On the other, they are motivated by the fact that capillary stresses have to be included into the equilibrium (reversible) part of the stress tensor. To that end, we shall introduce the subset of the constitutive state space $\mathcal{C}$, denoted by $\mathcal{C}_{s}$, from which the temperature $\theta$ is excluded 
\begin{equation} \label{CSS-subset_Korteweg}
	\mathcal{C}_{s} = \{ \rho, \rho_{,k}, \rho_{,kl}, \theta_{,k}, D_{kl} \}.   
\end{equation}
It is convenient to assume the specific entropy in the form which can easily be reduced to the equilibrium one. Mimicking \eqref{konstitutivne relacije s hipoteza 1 Ojlerovi fluidi}, it will be assumed that specific entropy $s$ is a function of the specific internal energy $e$, and the remaining constitutive quantities from $\mathcal{C}_{s}$. However, since specific internal energy is also determined by the constitutive relation \eqref{konstitutivne relacije 1 opste}$_{1}$, structure of the specific entropy constitutive function will be
\begin{equation} \label{konstitutivne relacije s hipoteza 1}
	s = s(e(\rho, \rho_{,k}, \rho_{,kl}, \theta_{,k}, D_{kl}, \theta), 
	\rho, \rho_{,k}, \rho_{,kl}, \theta_{,k}, D_{kl}). 
\end{equation} 
In addition, it will be assumed that thermodynamic relations between specific entropy $s$, specific internal energy $e$, mass density $\rho$ and thermodynamic pressure $p$ hold
\begin{equation} \label{konstitutivne relacije s hipoteza 2}
	\dfrac{\partial s}{\partial e} = \dfrac{1}{\theta}, \quad 
	\dfrac{\partial s}{\partial \rho} = - \dfrac{1}{\rho^{2}} \dfrac{p}{\theta}.
\end{equation}

Since specific entropy $s$ may depend on the constitutive variables explicitly, but also implicitly (through the specific internal energy $e$), we shall denote by an overline the total derivative of the specific entropy and compute it in the following way
\begin{equation} \label{totalni izvod}
	\overline{\dfrac{\partial s}{\partial\omega}} 
	= \dfrac{\partial s}{\partial e}\dfrac{\partial e}{\partial\omega} 
	+ \dfrac{\partial s}{\partial\omega}, \quad
	\overline{\dfrac{\partial s}{\partial \theta}} 
	= \dfrac{\partial s}{\partial e}\dfrac{\partial e}{\partial \theta},  
\end{equation}
where $\omega$ stands for any element of $\mathcal{C}_{s}$, $\omega \in \mathcal{C}_{s}$. 

\subsection{The method of multipliers} 

The essence of the method of multipliers lies in specific formulation of the entropy principle proposed by Liu. According to this approach, the entropy balance law is the main equation which determines the evolution of thermodynamic process, whereas the balance laws are the constraints that have to be satisfied. To transform the problem with constraints into the one without them, Lagrange multipliers $\Lambda^\rho$, $\Lambda^{v_i}$, $\Lambda^e$ and $\Lambda^{\nabla\rho}_i$ are introduced and the entropy balance law is written in extended form
\begin{align} \label{prosirena entropijska nejednakost Kortevegovi fluidi}
	\dfrac{\partial}{\partial t}\left(\rho s\right) & +\dfrac{\partial}{\partial x_j}\left(\rho sv_j+\varphi_j\right)
	-\Lambda^\rho\left[\dfrac{\partial\rho}{\partial t}+\dfrac{\partial}{\partial x_j}\left(\rho v_j\right)\right]
	\nonumber \\
	&-\Lambda^{v_i}\left[\dfrac{\partial}{\partial t}(\rho v_i)+\dfrac{\partial}{\partial x_j}\left(\rho v_i v_j -t_{ij}\right)\right]
	\\
	&-\Lambda^e\left[\frac{\partial}{\partial t}\left(\frac{1}{2}\rho\vert\mathbf{v}\vert ^2+\rho e\right)+\dfrac{\partial}{\partial x_j}\left(\left(\frac{1}{2}\rho\vert\mathbf{v}\vert ^2+\rho e\right)v_j-t_{ij}v_i+q_j\right)\right]
	\nonumber \\
	&-\Lambda^{\nabla\rho}_i\left[\dfrac{\partial}{\partial t}\rho_{,i}+\dfrac{\partial}{\partial x_j}(\rho_{,i}v_j+\rho (D_{ji}+W_{ji}))\right]=\Sigma\geq 0. 
	\nonumber
\end{align}
Taking into account the constitutive relations \eqref{konstitutivne relacije 1 opste} and \eqref{konstitutivne relacije s-phi opste}, and the hypothesis \eqref{konstitutivne relacije s hipoteza 1}, extended entropy balance law becomes
\begin{align*} %\label{produkcija entropije Kortevegovi fluidi}
	\dfrac{\partial\rho}{\partial t}s & + \rho\left(\overline{\dfrac{\partial s}{\partial\rho}}\dfrac{\partial\rho}{\partial t}+\overline{\dfrac{\partial s}{\partial\rho_{,k}}}\dfrac{\partial\rho_{,k}}{\partial t}+\overline{\dfrac{\partial s}{\partial\rho_{,kl}}}\dfrac{\partial\rho_{,kl}}{\partial t}+\overline{\dfrac{\partial s}{\partial\theta}}\dfrac{\partial\theta}{\partial t}+\overline{\dfrac{\partial s}{\partial\theta_{,k}}}\dfrac{\partial\theta_{,k}}{\partial t}+\overline{\dfrac{\partial s}{\partial D_{kl}}}\dfrac{\partial D_{kl}}{\partial t}\right)+\dfrac{\partial\rho}{\partial x_j}sv_j
	\\
	&+\rho\left(\overline{\dfrac{\partial s}{\partial\rho}}\dfrac{\partial\rho}{\partial x_j}+\overline{\dfrac{\partial s}{\partial\rho_{,k}}}\dfrac{\partial\rho_{,k}}{\partial x_j}+\overline{\dfrac{\partial s}{\partial\rho_{,kl}}}\dfrac{\partial\rho_{,kl}}{\partial x_j}+\overline{\dfrac{\partial s}{\partial\theta}}\dfrac{\partial\theta}{\partial x_j}+\overline{\dfrac{\partial s}{\partial\theta_{,k}}}\dfrac{\partial\theta_{,k}}{\partial x_j}+\overline{\dfrac{\partial s}{\partial D_{kl}}}\dfrac{\partial D_{kl}}{\partial x_j}\right)v_j+\rho s \dfrac{\partial v_j}{\partial x_j}
	\\
	&+\dfrac{\partial\varphi_j}{\partial\rho}\dfrac{\partial\rho}{\partial x_j}+\dfrac{\partial\varphi_j}{\partial\rho_{,k}}\dfrac{\partial\rho_{,k}}{\partial x_j}+\dfrac{\partial\varphi_j}{\partial\rho_{,kl}}\dfrac{\partial\rho_{,kl}}{\partial x_j}+\dfrac{\partial\varphi_j}{\partial\theta}\dfrac{\partial\theta}{\partial x_j}+\dfrac{\partial\varphi_j}{\partial\theta_{,k}}\dfrac{\partial\theta_{,k}}{\partial x_j}+\dfrac{\partial\varphi_j}{\partial D_{kl}}\dfrac{\partial D_{kl}}{\partial x_j} 
	\\
	&-\Lambda^\rho\left[\dfrac{\partial\rho}{\partial t}+\dfrac{\partial\rho}{\partial x_j}v_j+\rho\dfrac{\partial v_j}{\partial x_j}\right] 
	\\
	&-\Lambda^{v_i}\left[\dfrac{\partial\rho}{\partial t}v_i+\rho\dfrac{\partial v_i}{\partial t}+\dfrac{\partial\rho}{\partial x_j}v_iv_j+\rho\dfrac{\partial v_i}{\partial x_j}v_j+\rho v_i\dfrac{\partial v_j}{\partial x_j}\right.
	\\
	& \quad -\left.\left(\dfrac{\partial t_{ij}}{\partial\rho}\dfrac{\partial\rho}{\partial x_j}+\dfrac{\partial t_{ij}}{\partial\rho_{,k}}\dfrac{\partial\rho_{,k}}{\partial x_j}+\dfrac{\partial t_{ij}}{\partial\rho_{,kl}}\dfrac{\partial\rho_{,kl}}{\partial x_j}+\dfrac{\partial t_{ij}}{\partial\theta}\dfrac{\partial\theta}{\partial x_j}+\dfrac{\partial t_{ij}}{\partial\theta_{,k}}\dfrac{\partial\theta_{,k}}{\partial x_j}+\dfrac{\partial t_{ij}}{\partial D_{kl}}\dfrac{\partial D_{kl}}{\partial x_j}\right)\right] 
	\\
	&-\Lambda^e\left[\dfrac{1}{2}\dfrac{\partial\rho}{\partial t}|\mathbf{v}|^2+\rho\dfrac{\partial v_i}{\partial t}v_i+\dfrac{\partial\rho}{\partial t} e\right. 
	\\
	& \quad +\rho\left(\dfrac{\partial e}{\partial\rho}\dfrac{\partial\rho}{\partial t}+\dfrac{\partial e}{\partial\rho_{,k}}\dfrac{\partial\rho_{,k}}{\partial t}+\dfrac{\partial e}{\partial\rho_{,kl}}\dfrac{\partial\rho_{,kl}}{\partial t}+\dfrac{\partial e}{\partial\theta}\dfrac{\partial\theta}{\partial t}+\dfrac{\partial e}{\partial\theta_{,k}}\dfrac{\partial\theta_{,k}}{\partial t}+\dfrac{\partial e}{\partial D_{kl}}\dfrac{\partial D_{kl}}{\partial t}\right) 
	\\
	& \quad +v_j\left(\dfrac{1}{2}\dfrac{\partial\rho}{\partial x_j}|\mathbf{v}|^2+\rho\dfrac{\partial v_i}{\partial x_j}v_i+\dfrac{\partial\rho}{\partial x_j} e\right. 
	\\
	& \quad \left.+\rho\left(\dfrac{\partial e}{\partial\rho}\dfrac{\partial\rho}{\partial x_j}+\dfrac{\partial e}{\partial\rho_{,k}}\dfrac{\partial\rho_{,k}}{\partial x_j}+\dfrac{\partial e}{\partial\rho_{,kl}}\dfrac{\partial\rho_{,kl}}{\partial x_j}+\dfrac{\partial e}{\partial\theta}\dfrac{\partial\theta}{\partial x_j}+\dfrac{\partial e}{\partial\theta_{,k}}\dfrac{\partial\theta_{,k}}{\partial x_j}+\dfrac{\partial e}{\partial D_{kl}}\dfrac{\partial D_{kl}}{\partial x_j}\right)\right) 
	\\
	& \quad +\left(\dfrac{1}{2}\rho\vert\mathbf{v}\vert^2+\rho e\right)\dfrac{\partial v_j}{\partial x_j} 
	\\
	& \quad -\left(\dfrac{\partial t_{ij}}{\partial\rho}\dfrac{\partial\rho}{\partial x_j}+\dfrac{\partial t_{ij}}{\partial\rho_{,k}}\dfrac{\partial\rho_{,k}}{\partial x_j}+\dfrac{\partial t_{ij}}{\partial\rho_{,kl}}\dfrac{\partial\rho_{,kl}}{\partial x_j}+\dfrac{\partial t_{ij}}{\partial\theta}\dfrac{\partial\theta}{\partial x_j}+\dfrac{\partial t_{ij}}{\partial\theta_{,k}}\dfrac{\partial\theta_{,k}}{\partial x_j}+\dfrac{\partial t_{ij}}{\partial D_{kl}}\dfrac{\partial D_{kl}}{\partial x_j}\right)v_i-t_{ij}\dfrac{\partial v_i}{\partial x_j} 
	\\
	& \quad \left.+\dfrac{\partial q_j}{\partial\rho}\dfrac{\partial\rho}{\partial x_j}+\dfrac{\partial q_j}{\partial\rho_{,k}}\dfrac{\partial\rho_{,k}}{\partial x_j}+\dfrac{\partial q_j}{\partial\rho_{,kl}}\dfrac{\partial\rho_{,kl}}{\partial x_j}+\dfrac{\partial q_j}{\partial\theta}\dfrac{\partial\theta}{\partial x_j}+\dfrac{\partial q_j}{\partial\theta_{,k}}\dfrac{\partial\theta_{,k}}{\partial x_j}+\dfrac{\partial q_j}{\partial D_{kl}}\dfrac{\partial D_{kl}}{\partial x_j}\right] 
	\\
	&-\Lambda^{\nabla\rho}_i\left[\dfrac{\partial\rho_{,i}}{\partial t}+\dfrac{\partial\rho_{,i}}{\partial x_j}v_j+\rho_{,i}\dfrac{\partial v_j}{\partial x_j}+\dfrac{\partial\rho}{\partial x_j}D_{ji}+\rho\dfrac{\partial D_{ji}}{\partial x_j}+\dfrac{\partial\rho}{\partial x_j}W_{ji}+\rho\dfrac{\partial W_{ji}}{\partial x_j}\right]=\Sigma\geq 0.
\end{align*}
Substituting the identities \eqref{App:Lij} and \eqref{App:divergenceLij} into the last equation, we obtain the new form of the extended entropy balance law
\begin{align} \label{rezidualna nejednakost Kortevegovi fluidi}
	\dfrac{\partial\rho}{\partial t} & \left[s+\rho\overline{\dfrac{\partial s}{\partial\rho}}-\Lambda^\rho-\Lambda^{v_i}v_i-\dfrac{1}{2}\Lambda^e|\mathbf{v}|^2-\Lambda^ e e-\Lambda^e\rho\dfrac{\partial e}{\partial\rho}\right] 
	\nonumber \\
	& + \dfrac{\partial\rho_{,k}}{\partial t}\left[\rho\overline{\dfrac{\partial s}{\partial\rho_{,k}}}-\Lambda^e\rho\dfrac{\partial e}{\partial\rho_{,k}}-\Lambda^{\nabla\rho}_k\right] 
	\nonumber \\ 
	& +\dfrac{\partial\rho_{,kl}}{\partial t}\left[\rho\overline{\dfrac{\partial s}{\partial\rho_{,kl}}}-\Lambda^e\rho\dfrac{\partial e}{\partial\rho_{,kl}}\right] 
	\nonumber \\
	&+\dfrac{\partial v_i}{\partial t}\left[-\Lambda^{v_i}\rho-\Lambda^e\rho v_i\right] 
	\nonumber \\
	&+\dfrac{\partial\theta}{\partial t}\left[\rho\overline{\dfrac{\partial s}{\partial\theta}}-\Lambda^e\rho\dfrac{\partial e}{\partial\theta}\right] 
	\nonumber \\ 
	& +\dfrac{\partial\theta_{,k}}{\partial t}\left[\rho\overline{\dfrac{\partial s}{\partial\theta_{,k}}}-\Lambda^e\rho\dfrac{\partial e}{\partial\theta_{,k}}\right] 
	\\
	&+\dfrac{\partial D_{kl}}{\partial t}\left[\rho\overline{\dfrac{\partial s}{\partial D_{kl}}}-\Lambda^e\rho\dfrac{\partial e}{\partial D_{kl}}\right] 
	\nonumber \\
	& +\dfrac{\partial\rho}{\partial x_j}\left[sv_j+\rho\overline{\dfrac{\partial s}{\partial\rho}}v_j+\dfrac{\partial\varphi_j}{\partial\rho}-\Lambda^\rho v_j-\Lambda^{v_i}v_iv_j+\Lambda^{v_i}\dfrac{\partial t_{ij}}{\partial\rho}
	-\dfrac{1}{2}\Lambda^e|\mathbf{v}|^2v_j \right. 
	\nonumber \\ 
	& \quad \left. %\quad \left.
	-\Lambda^e e v_j-\Lambda^e \rho\dfrac{\partial e}{\partial\rho}v_j+\Lambda^e\dfrac{\partial t_{ij}}{\partial\rho}v_i-\Lambda^e\dfrac{\partial q_j}{\partial\rho}\right]
	\nonumber \\
	&+\dfrac{\partial\rho_{,k}}{\partial x_j}\left[\rho\overline{\dfrac{\partial s}{\partial\rho_{,k}}}v_j+\dfrac{\partial\varphi_j}{\partial\rho_{,k}}+\Lambda^{v_i}\dfrac{\partial t_{ij}}{\partial\rho_{,k}}-\Lambda^e\rho\dfrac{\partial e}{\partial\rho_{,k}}v_j %\right. 
	% \nonumber \\ & \quad \left. 
	+ \Lambda^e\dfrac{\partial t_{ij}}{\partial\rho_{,k}}v_i-\Lambda^e\dfrac{\partial q_j}{\partial\rho_{,k}}-\Lambda^{\nabla\rho}_k v_j\right] 
	\nonumber \\
	&+\dfrac{\partial\rho_{,kl}}{\partial x_j}\left[\rho\overline{\dfrac{\partial s}{\partial\rho_{,kl}}}v_j+\dfrac{\partial\varphi_j}{\partial\rho_{,kl}}+\Lambda^{v_i}\dfrac{\partial t_{ij}}{\partial\rho_{,kl}}-\Lambda^e\rho\dfrac{\partial e}{\partial\rho_{,kl}}v_j %\right. 
	% \nonumber \\ & \quad \left. 
	+\Lambda^e\dfrac{\partial t_{ij}}{\partial\rho_{,kl}}v_i-\Lambda^e\dfrac{\partial q_j}{\partial\rho_{,kl}}\right] 
	\nonumber \\
	&+\dfrac{\partial\theta}{\partial x_j}\left[\rho\overline{\dfrac{\partial s}{\partial\theta}}v_j+\dfrac{\partial\varphi_j}{\partial\theta}+\Lambda^{v_i}\dfrac{\partial t_{ij}}{\partial\theta}-\Lambda^e\rho\dfrac{\partial e}{\partial\theta}v_j+\Lambda^e\dfrac{\partial t_{ij}}{\partial\theta}v_i-\Lambda^e\dfrac{\partial q_j}{\partial\theta}\right] 
	\nonumber \\
	&+\dfrac{\partial\theta_{,k}}{\partial x_j}\left[\rho\overline{\dfrac{\partial s}{\partial\theta_{,k}}}v_j+\dfrac{\partial\varphi_j}{\partial\theta_{,k}}+\Lambda^{v_i}\dfrac{\partial t_{ij}}{\partial\theta_{,k}}-\Lambda^e\rho\dfrac{\partial e}{\partial\theta_{,k}}v_j+\Lambda^e\dfrac{\partial t_{ij}}{\partial\theta_{,k}}v_i-\Lambda^e\dfrac{\partial q_j}{\partial\theta_{,k}}\right] 
	\nonumber \\
	&+\dfrac{\partial D_{kl}}{\partial x_j}\left[\rho\overline{\dfrac{\partial s}{\partial D_{kl}}}v_j+\dfrac{\partial\varphi_j}{\partial D_{kl}}+\Lambda^{v_i}\dfrac{\partial t_{ij}}{\partial D_{kl}}-\Lambda^e\rho\dfrac{\partial e}{\partial D_{kl}}v_j % \right. 
	% \nonumber \\ & \quad \left. 
	+\Lambda^e\dfrac{\partial t_{ij}}{\partial D_{kl}}v_i-\Lambda^e\dfrac{\partial q_j}{\partial D_{kl}}-\Lambda^{\nabla\rho}_j\rho\delta_{kl}\right] 
	\nonumber \\
	&+D_{ij}\left[\rho s \delta_{ij}-\Lambda^\rho\rho\delta_{ij}-\Lambda^{v_i}\rho v_j-\Lambda^{v_k}\rho v_k\delta_{ij}-\Lambda^e\rho v_i v_j-\dfrac{1}{2}\rho\vert \mathbf{v}\vert^2\Lambda^e\delta_{ij} \right. 
	\nonumber \\ 
	& \quad \left. -\Lambda^e\rho e\delta_{ij}+\Lambda^et_{ij}-\Lambda^{\nabla\rho}_k\dfrac{\partial\rho}{\partial x_k}\delta_{ij}-\Lambda^{\nabla\rho}_j\dfrac{\partial\rho}{\partial x_i}\right] 
	\nonumber \\
	&+W_{ij}\left[-\Lambda^{v_i}\rho v_j-\Lambda^e\rho v_iv_j+\Lambda^e t_{ij}-\Lambda^{\nabla\rho}_j\dfrac{\partial\rho}{\partial x_i}\right]=\Sigma\geq 0. 
	\nonumber 
\end{align}

In equation \eqref{rezidualna nejednakost Kortevegovi fluidi} the differential part is expressed as the sum of monomials, consisted of the derivatives of field variables multiplied by coefficients which contain Lagrange multipliers and constitutive functions. If the derivative of a field variable does not belong to the constitutive state space, the sign of its monomial cannot be controlled and may lead to violation of the entropy inequality. In that case, the corresponding coefficient have to be equal zero. This reasoning leads to identification of the process direction vector
\begin{equation} \label{ProcessDirectionVector-Korteweg}
	\mathbf{X} = \left( \dfrac{\partial \rho}{\partial t}, \dfrac{\partial\rho_{,k}}{\partial t}, \dfrac{\partial\rho_{,kl}}{\partial t}, \dfrac{\partial v_i}{\partial t}, 
	\dfrac{\partial \theta}{\partial t}, \dfrac{\partial \theta_{,k}}{\partial t}, \dfrac{\partial D_{kl}}{\partial t}, \dfrac{\partial \rho_{,kl}}{\partial x_j}, 
	\dfrac{\partial \theta_{,k}}{\partial x_j}, \dfrac{\partial D_{kl}}{\partial x_{j}} \right), 
\end{equation}
and facilitate application of the Lemma \ref{lemma:Liu-Lemma}. As a consequence, extended entropy balance law is split into a set of equations, whereas the remaining term will represent the \emph{residual inequality}. 

First group of equations consists of the coefficients multiplying the time derivatives of the field variables, $\dfrac{\partial\rho}{\partial t}$, $\dfrac{\partial\rho_{,k}}{\partial t}$, $\dfrac{\partial\rho_{,kl}}{\partial t}$, $\dfrac{\partial v_i}{\partial t}$, $\dfrac{\partial\theta}{\partial t}$, $\dfrac{\partial\theta{,_k}}{\partial t}$ and $\dfrac{\partial D_{kl}}{\partial t}$: 
\begin{align}
	\label{Trho Kortevegovi fluidi}
	& s+\rho\overline{\dfrac{\partial s}{\partial\rho}}-\Lambda^\rho-\Lambda^{v_i}v_i-\dfrac{1}{2}\Lambda^e|\mathbf{v}|^2-\Lambda^ e e-\Lambda^e\rho\dfrac{\partial e}{\partial\rho}=0, 
	\\
	\label{Trho_,k Kortevegovi fluidi}
	& \rho\overline{\dfrac{\partial s}{\partial\rho_{,k}}}-\Lambda^e\rho\dfrac{\partial e}{\partial\rho_{,k}}-\Lambda^{\nabla\rho}_k=0, 
	\\
	& \label{Trho_,kl Kortevegovi fluidi}
	\rho\overline{\dfrac{\partial s}{\partial\rho_{,kl}}}-\Lambda^e\rho\dfrac{\partial e}{\partial\rho_{,kl}}=0,
	\\
	\label{Tv_i Kortevegovi fluidi}
	& -\Lambda^{v_i}\rho-\Lambda^e\rho v_i=0,\quad i=1,2,3. 
	\\
	\label{Ttheta Kortevegovi fluidi}
	& \rho\overline{\dfrac{\partial s}{\partial\theta}}-\Lambda^e\rho\dfrac{\partial e}{\partial\theta}=0, 
	\\
	\label{Ttheta_,k Kortevegovi fluidi}
	& \rho\overline{\dfrac{\partial s}{\partial\theta_{,k}}}-\Lambda^e\rho\dfrac{\partial e}{\partial\theta_{,k}}=0, \quad k=1,2,3,
	\\
	\label{TD_kl Kortevegovi fluidi}
	& \rho\overline{\dfrac{\partial s}{\partial D_{kl}}}-\Lambda^e\rho\dfrac{\partial e}{\partial D_{kl}}=0, \quad k=1,2,3, \quad l=1,2,3.
\end{align}
Second group of equations consists of the coefficients multiplying the space derivatives of the field variables, $\dfrac{\partial\rho_{,kl}}{\partial x_j}$, $\dfrac{\partial\theta_{,k}}{\partial x_j}$ and $\dfrac{\partial D_{kl}}{\partial x_j}$:  
\begin{align}
	\label{X_jrho_,kl Kortevegovi fluidi}
	& \rho\overline{\dfrac{\partial s}{\partial\rho_{,kl}}}v_j+\dfrac{\partial\varphi_j}{\partial\rho_{,kl}}+\Lambda^{v_i}\dfrac{\partial t_{ij}}{\partial\rho_{,kl}}-\Lambda^e\rho\dfrac{\partial e}{\partial\rho_{,kl}}v_j+\Lambda^e\dfrac{\partial t_{ij}}{\partial\rho_{,kl}}v_i-\Lambda^e\dfrac{\partial q_j}{\partial\rho_{,kl}}=0, 
	\\
	\label{X_jtheta_,k Kortevegovi fluidi}
	& \rho\overline{\dfrac{\partial s}{\partial\theta_{,k}}}v_j+\dfrac{\partial\varphi_j}{\partial\theta_{,k}}+\Lambda^{v_i}\dfrac{\partial t_{ij}}{\partial\theta_{,k}}-\Lambda^e\rho\dfrac{\partial e}{\partial\theta_{,k}}v_j+\Lambda^e\dfrac{\partial t_{ij}}{\partial\theta_{,k}}v_i-\Lambda^e\dfrac{\partial q_j}{\partial\theta_{,k}}=0, 
	\\
	\label{X_jD_kl Kortevegovi fluidi}
	& \rho\overline{\dfrac{\partial s}{\partial D_{kl}}}v_j+\dfrac{\partial\varphi_j}{\partial D_{kl}}+\Lambda^{v_i}\dfrac{\partial t_{ij}}{\partial D_{kl}}-\Lambda^e\rho\dfrac{\partial e}{\partial D_{kl}}v_j+\Lambda^e\dfrac{\partial t_{ij}}{\partial D_{kl}}v_i-\Lambda^e\dfrac{\partial q_j}{\partial D_{kl}}-\Lambda^{\nabla\rho}_j\rho\delta_{kl}=0.
\end{align}
Exception from this rule is the monomial which contains the vorticity tensor $W_{ij}$, which is skew-symmetric, and in this form neither belongs to the constitutive state space $\mathcal{C}$, nor to the process direction vector. Therefore, it will be required that whole monomial is zero
\begin{equation} \label{Wij Kortevegovi fluidi}
	\left[-\Lambda^{v_i}\rho v_j-\Lambda^e\rho v_iv_j+\Lambda^e t_{ij}-\Lambda^{\nabla\rho}_j\dfrac{\partial\rho}{\partial x_i}\right]W_{ij}=0.
\end{equation}
This will turn out to be crucial for the computation of multipliers. Finally, taking into account equations \eqref{Trho Kortevegovi fluidi}--\eqref{Wij Kortevegovi fluidi}, the remaining part of equation \eqref{rezidualna nejednakost Kortevegovi fluidi} will represent the basic form of the residual inequality
\begin{align} \label{osnovna rezidualna nejednakost Kortevegovi fluidi}
	& \dfrac{\partial\rho}{\partial x_j} \left[sv_j+\rho\overline{\dfrac{\partial s}{\partial\rho}}v_j+\dfrac{\partial\varphi_j}{\partial\rho}-\Lambda^\rho v_j-\Lambda^{v_i}v_iv_j+\Lambda^{v_i}\dfrac{\partial t_{ij}}{\partial\rho}
	\right. 
	\nonumber \\  
	& \qquad \left. -\dfrac{1}{2}\Lambda^e|\mathbf{v}|^2v_j -\Lambda^e e v_j-\Lambda^e \rho\dfrac{\partial e}{\partial\rho}v_j+\Lambda^e\dfrac{\partial t_{ij}}{\partial\rho}v_i-\Lambda^e\dfrac{\partial q_j}{\partial\rho}\right] 
	\nonumber \\
	& \quad + \dfrac{\partial\rho_{,k}}{\partial x_j} \left[\rho\overline{\dfrac{\partial s}{\partial\rho_{,k}}}v_j+\dfrac{\partial\varphi_j}{\partial\rho_{,k}}+\Lambda^{v_i}\dfrac{\partial t_{ij}}{\partial\rho_{,k}}-\Lambda^e\rho\dfrac{\partial e}{\partial\rho_{,k}}v_j % \right. 
	% \\ & \quad \left. 
	+ \Lambda^e\dfrac{\partial t_{ij}}{\partial\rho_{,k}}v_i-\Lambda^e\dfrac{\partial q_j}{\partial\rho_{,k}}-\Lambda^{\nabla\rho}_k v_j\right] 
	\\
	& \quad +\dfrac{\partial\theta}{\partial x_j}\left[\rho\overline{\dfrac{\partial s}{\partial\theta}}v_j+\dfrac{\partial\varphi_j}{\partial\theta}+\Lambda^{v_i}\dfrac{\partial t_{ij}}{\partial\theta}-\Lambda^e\rho\dfrac{\partial e}{\partial\theta}v_j+\Lambda^e\dfrac{\partial t_{ij}}{\partial\theta}v_i-\Lambda^e\dfrac{\partial q_j}{\partial\theta}\right] 
	\nonumber \\
	& \quad + D_{ij}\left[\rho s \delta_{ij}-\Lambda^\rho\rho\delta_{ij}-\Lambda^{v_i}\rho v_j-\Lambda^{v_k}\rho v_k\delta_{ij}-\Lambda^e\rho v_i v_j-\dfrac{1}{2}\rho\vert \mathbf{v}\vert^2\Lambda^e\delta_{ij} \right.
	\nonumber \\ 
	& \qquad \left. - \Lambda^e\rho e\delta_{ij} + \Lambda^et_{ij}-\Lambda^{\nabla\rho}_k\dfrac{\partial\rho}{\partial x_k}\delta_{ij}-\Lambda^{\nabla\rho}_j\dfrac{\partial\rho}{\partial x_i} \right] = \Sigma\geq 0. 
	\nonumber 
\end{align}

\subsection{Computation of the multipliers}

The goal of the method of multipliers is to determine the constitutive relations that are compatible with entropy (residual) inequality. To achieve this goal, one has to determine the multipliers at the same time. Therefore, we shall at first analyze the equations obtained from the entropy balance law, with the aim of computation of the unknown multipliers. 

\begin{lemma} \label{Lema:Lambda1}
	For the multipliers $\Lambda^{\rho}$ and $\Lambda^{v_i}$ the following relations hold 
	\begin{align}
		\label{modifikovana Trho Kortevegovi fluidi}
		& \Lambda^\rho = s + \rho\overline{\dfrac{\partial s}{\partial\rho}} 
		+ \dfrac{1}{2}\Lambda^e|\mathbf{v}|^2 - \Lambda^ e e 
		- \Lambda^e \rho \dfrac{\partial e}{\partial\rho},
		\\
		\label{Tv_i' Kortevegovi fluidi}
		& \Lambda^{v_i}=-\Lambda^e v_i, \quad i = 1, 2, 3.
	\end{align}
\end{lemma}

\begin{proof}
	Since $\rho > 0$, equation \eqref{Tv_i' Kortevegovi fluidi} stems directly from \eqref{Tv_i Kortevegovi fluidi}. Insertion of \eqref{Tv_i' Kortevegovi fluidi} into \eqref{Trho Kortevegovi fluidi} yields \eqref{modifikovana Trho Kortevegovi fluidi}. 
\end{proof}

The result of Lemma \ref{Lema:Lambda1} is only preliminary. It does not provide the final form of the multipliers, but it is necessary for further analysis. In the sequel, we shall firstly determine the multipliers $\Lambda^{e}$, $\Lambda^{v_i}$ i $\Lambda^{\rho}$, and in the end $\Lambda^{\nabla\rho}_j$. 
\begin{theorem} \label{Te:Lambda-e}
	The multiplier $\Lambda^{e}$ of the energy conservation law \eqref{zakon održanja energije Kortevegovi fluidi} reads
	\begin{equation} \label{mnozitelj Lambda-e Kortevegovi fluidi}
		\Lambda^{e} = \dfrac{1}{\theta}. 
	\end{equation}
	Besides, the following holds
	\begin{equation} \label{s-argumenti Kortevegovi fluidi}
		\dfrac{\partial s}{\partial \rho_{,kl}} = 0, \quad 
		% \dfrac{\partial s}{\partial \theta} = 0, \quad
		\dfrac{\partial s}{\partial \theta_{,k}} = 0, \quad
		\dfrac{\partial s}{\partial D_{kl}} = 0.
	\end{equation}
\end{theorem}

\begin{proof}
	Using the definition of the total derivative \eqref{totalni izvod}, from equation \eqref{Trho_,kl Kortevegovi fluidi} it follows
	\begin{equation*}
		\dfrac{\partial s}{\partial e}\dfrac{\partial e}{\partial\rho_{,kl}}+\dfrac{\partial s}{\partial\rho_{,kl}}-\Lambda^e\dfrac{\partial e}{\partial\rho_{,kl}}=0,
	\end{equation*}
	or, equivalently
	\begin{equation} \label{s po rho-kl}
		\left(\dfrac{\partial s}{\partial e}-\Lambda^e\right)\dfrac{\partial e}{\partial\rho_{,kl}}+\dfrac{\partial s}{\partial\rho_{,kl}}=0.
	\end{equation} 
	Similarly, it can be shown that from equations \eqref{Ttheta Kortevegovi fluidi}, \eqref{Ttheta_,k Kortevegovi fluidi} and \eqref{TD_kl Kortevegovi fluidi} it follows
	\begin{align} 
		\label{s po theta}
		& \left(\dfrac{\partial s}{\partial e}-\Lambda^e\right)\dfrac{\partial e}{\partial\theta}=0, 
		\\
		% +\dfrac{\partial s}{\partial\theta}
		\label{s po theta-k}
		& \left(\dfrac{\partial s}{\partial e}-\Lambda^e\right)\dfrac{\partial e}{\partial\theta_{,k}}+\dfrac{\partial s}{\partial\theta_{,k}}=0, 
		\\
		\label{s po D-kl} 
		& \left(\dfrac{\partial s}{\partial e}-\Lambda^e\right)\dfrac{\partial e}{\partial D_{kl}}+\dfrac{\partial s}{\partial D_{kl}}=0,
	\end{align}
	respectively. Since $\dfrac{\partial e}{\partial \theta} \neq 0$ in general, it follows from \eqref{s po theta} 
	\begin{equation} \label{Lambda_e Kortevegovi fluidi}
		\Lambda^e = \dfrac{\partial s}{\partial e} \overset{\eqref{konstitutivne relacije s hipoteza 2}_{1}}{=} \dfrac{1}{\theta}, 
	\end{equation} 
	from equations \eqref{s po rho-kl}--\eqref{s po D-kl} follows directly \eqref{s-argumenti Kortevegovi fluidi}, which ends the proof.
\end{proof} 

\begin{corollary}
	The multipliers $\Lambda^{\rho}$ and $\Lambda^{v_i}$ of the mass conservation law \eqref{zakon održanja mase Kortevegovi fluidi} and the momentum conservation law \eqref{zakon održanja količine kretanja Kortevegovi fluidi}, respectively, are determined by the following relations 
	\begin{align} 
		\label{mnozitelj Lambda-rho Kortevegovi fluidi}
		& \Lambda^{\rho} = - \dfrac{1}{\theta} \left( e - \theta s + \frac{p}{\rho} 
		- \frac{1}{2} |\mathbf{v}|^{2} \right), 
		\\
		\label{mnozitelj Lambda-vi Kortevegovi fluidi}
		& \Lambda^{v_i} = - \dfrac{v_{i}}{\theta}, \quad i = 1, 2, 3.
	\end{align}
\end{corollary}
\begin{proof}
	Using the definition of the total derivative \eqref{totalni izvod} we have
	\begin{align*}
		\Lambda^\rho &\overset{\eqref{modifikovana Trho Kortevegovi fluidi}}{=} s + \rho\overline{\dfrac{\partial s}{\partial\rho}} 
		+ \dfrac{1}{2}\Lambda^e|\mathbf{v}|^2 - \Lambda^ e e 
		- \Lambda^e \rho \dfrac{\partial e}{\partial\rho}
		\\
		&\overset{\eqref{Lambda_e Kortevegovi fluidi}}{=} s + \rho\left(\dfrac{\partial s}{\partial e}\dfrac{\partial e}{\partial\rho}+\dfrac{\partial s}{\partial\rho}\right)+\dfrac{1}{2}\dfrac{\partial s}{\partial e}|\mathbf{v}|^2 -\dfrac{\partial s}{\partial e} e - \dfrac{\partial s}{\partial e} \rho \dfrac{\partial e}{\partial\rho}.	   
	\end{align*}
	Insertion of \eqref{mnozitelj Lambda-e Kortevegovi fluidi} and \eqref{konstitutivne relacije s hipoteza 2} into the last equation yields \eqref{mnozitelj Lambda-rho Kortevegovi fluidi}. Equation \eqref{mnozitelj Lambda-vi Kortevegovi fluidi} stems from \eqref{Tv_i' Kortevegovi fluidi} by insertion of \eqref{mnozitelj Lambda-e Kortevegovi fluidi}. 
\end{proof}

\begin{theorem} \label{Te:Lambda-grad-rho}
	The multiplier $\Lambda^{\nabla\rho}_i$ of the additional balance law \eqref{gradijent odrzanja mase} reads
	\begin{equation} \label{mnozitelj Lambda-grad-rho Kortevegovi fluidi}
		\Lambda^{\nabla\rho}_i = \dfrac{1}{\theta} \alpha(\rho,\theta) 
		\dfrac{\partial \rho}{\partial x_i}, \quad i = 1, 2, 3. 
	\end{equation}
\end{theorem} 

\begin{proof}
	Equation \eqref{Wij Kortevegovi fluidi} written in expanded form reads 
	\begin{equation} \label{WijExpand Kortevegovi fluidi}
		- \rho \left( \Lambda^{v_i} + \Lambda^e v_i \right) v_j W_{ij} + 
		\Lambda^e t_{ij} W_{ij} 
		- \Lambda^{\nabla\rho}_j \dfrac{\partial\rho}{\partial x_i} W_{ij} = 0.
	\end{equation} 
	Using \eqref{Tv_i' Kortevegovi fluidi} the first term vanishes. Since the stress tensor $\mathbf{t}={[t_{ij}]}_{i,j=1}^{3}$ is symmetric, and the vorticity tensor ${\mathbf{W}=[W_{ij}]}_{i,j=1}^{3}$ is skew-symmetric, their inner product will vanish, $t_{ij}W_{ij}=0$. In such a way \eqref{WijExpand Kortevegovi fluidi} reduces to 
	\begin{equation} \label{WijReduced Kortevegovi fluidi}
		\Lambda^{\nabla\rho}_j\dfrac{\partial\rho}{\partial x_i}W_{ij}=0, 
	\end{equation}
	which is the inner product of the vorticity tensor $W_{ij}$ with tensor product of the multiplier $\Lambda^{\nabla\rho}_j$ and the mass density gradient $\partial \rho/\partial x_{i}$. To satisfy \eqref{WijReduced Kortevegovi fluidi}, it is sufficient that the following tensor (i.e. tensor product) is symmetric
\begin{equation} \label{Lambda-grad-rho simetrija}
	\Lambda^{\nabla\rho}_j\dfrac{\partial\rho}{\partial x_i} 
	= \Lambda^{\nabla\rho}_i\dfrac{\partial\rho}{\partial x_j}, 
	\quad i,j = 1,2,3. 
\end{equation}
The simplest way to satisfy the condition \eqref{Lambda-grad-rho simetrija} is to choose the multiplier in the following form
\begin{equation} \label{Lambda-grad-rho konacno}
	\Lambda^{\nabla\rho}_i = \dfrac{1}{\theta} \alpha(\rho, \theta) 
	\dfrac{\partial\rho}{\partial x_i}
	\quad i = 1,2,3,
\end{equation} 
where $\alpha(\rho,\theta)$ is an arbitrary scalar function of the given arguments.
\end{proof}

Let us mention that the factor $1/\theta$ in equation \eqref{Lambda-grad-rho konacno} is chosen for convenience, which will become apparent in the subsequent analysis. Also, more complex structure of the scalar function in \eqref{Lambda-grad-rho konacno}, which may include other quantities from the constitutive state space, would not have a significant contribution to the construction of constitutive relations for the Korteweg fluids. For that reason, we have restricted it to $\alpha(\rho,\theta)$. 

The temperature dependence of $\alpha$ is not merely a formal generalization: it is supported by kinetic-theory derivations of \cite{bhattacharjee2024temperature} from the Enskog--Vlasov equation, where the capillary coefficient inherits temperature dependence through the pair correlation function. If $\alpha$ does not depend on $\theta$, the cross-coupling term $\alpha - \theta\,\partial\alpha/\partial\theta$ in the non-equilibrium constitutive relations (Theorem \ref{Te:Neravnotezni Protoci}) vanishes, and both cases reduce to the standard viscous Korteweg fluid. 

\subsection{Analysis of the residual inequality}

In this Section we shall thoroughly analyze the residual inequality \eqref{osnovna rezidualna nejednakost Kortevegovi fluidi}. At first, it will be transformed using the results of the previous Section. After that, it will be reduced to the form appropriate for the study of constitutive relations. The following Lemma expresses the result required for transformation of the residual inequality. 

\begin{lemma} \label{Lema:phi_j jednakosti}
	The following relations hold 
	\begin{align}
		\label{X_jrho_,kl prim Kortevegovi fluidi}
		& \dfrac{\partial\varphi_j}{\partial\rho_{,kl}}-\Lambda^e\dfrac{\partial q_j}{\partial\rho_{,kl}}=0, 
		\\
		\label{X_jtheta_,k prim Kortevegovi fluidi}
		& \dfrac{\partial\varphi_j}{\partial\theta_{,k}}-\Lambda^e\dfrac{\partial q_j}{\partial\theta_{,k}}=0, 
		\\
		\label{X_jD_kl prim Kortevegovi fluidi}
		& \dfrac{\partial\varphi_j}{\partial D_{kl}}-\Lambda^e\dfrac{\partial q_j}{\partial D_{kl}}-\Lambda^{\nabla\rho}_j\rho\delta_{kl}=0.
	\end{align}
\end{lemma}

\begin{proof}
	Equations \eqref{X_jrho_,kl prim Kortevegovi fluidi}--\eqref{X_jD_kl prim Kortevegovi fluidi} follow from equations \eqref{X_jrho_,kl Kortevegovi fluidi}--\eqref{X_jD_kl Kortevegovi fluidi}. In particular, equation \eqref{X_jrho_,kl Kortevegovi fluidi} can be written as
	\begin{equation*}
		\label{koeficijent uz rho_,kl po x_j Kortevegovi fluidi}
		\left( \rho\overline{\dfrac{\partial s}{\partial\rho_{,kl}}} 
		- \Lambda^e\rho\dfrac{\partial e}{\partial\rho_{,kl}} \right) v_j 
		+ \dfrac{\partial\varphi_j}{\partial\rho_{,kl}} 
		+ \left( \Lambda^{v_i} + \Lambda^e v_i \right) \dfrac{\partial t_{ij}}{\partial\rho_{,kl}} 
		- \Lambda^e\dfrac{\partial q_j}{\partial\rho_{,kl}} = 0,
	\end{equation*}
	where using \eqref{Trho_,kl Kortevegovi fluidi} and \eqref{Tv_i' Kortevegovi fluidi} we obtain \eqref{X_jrho_,kl prim Kortevegovi fluidi}. Similarly, equation \eqref{X_jtheta_,k Kortevegovi fluidi} can be transformed into
	\begin{equation*}
		\label{koeficijent uz theta_,k po x_j Kortevegovi fluidi}
		\left(\rho\overline{\dfrac{\partial s}{\partial\theta_{,k}}} 
		- \Lambda^e\rho\dfrac{\partial e}{\partial\theta_{,k}}\right) v_j 
		+ \dfrac{\partial\varphi_j}{\partial\theta_{,k}} 
		+ \left( \Lambda^{v_i} + \Lambda^e v_i \right) \dfrac{\partial t_{ij}}{\partial\theta_{,k}} 
		- \Lambda^e\dfrac{\partial q_j}{\partial\theta_{,k}} = 0, 
	\end{equation*}
	and applying \eqref{Ttheta_,k Kortevegovi fluidi} and \eqref{Tv_i' Kortevegovi fluidi} can be reduced to \eqref{X_jtheta_,k prim Kortevegovi fluidi}. Finally, equation \eqref{X_jD_kl Kortevegovi fluidi} after transformation reads 
	\begin{equation*}
		\label{koeficijent uz D_kl po x_j Kortevegovi fluidi}
		\left(\rho\overline{\dfrac{\partial s}{\partial D_{kl}}} 
		- \Lambda^e\rho\dfrac{\partial e}{\partial D_{kl}}\right) v_j 
		+ \dfrac{\partial\varphi_j}{\partial D_{kl}} 
		+ \left( \Lambda^{v_i} + \Lambda^e v_i \right) \dfrac{\partial t_{ij}}{\partial D_{kl}} 
		- \Lambda^e\dfrac{\partial q_j}{\partial D_{kl}} 
		- \Lambda^{\nabla\rho}_j\rho\delta_{kl} = 0,
	\end{equation*}
	where application of \eqref{TD_kl Kortevegovi fluidi} and \eqref{Tv_i' Kortevegovi fluidi} leads to \eqref{X_jD_kl prim Kortevegovi fluidi}. 
\end{proof}

Using Lemma \ref{Lema:phi_j jednakosti}, the residual inequality can be simplified. 

\begin{lemma} \label{Lema:Rezidualna nejednakost 2}
	Residual inequality \eqref{osnovna rezidualna nejednakost Kortevegovi fluidi} reduces to
	\begin{align} \label{rezidualna nejednakost 2 Kortevegovi fluidi}
		& \dfrac{\partial\rho}{\partial x_j} \left[\dfrac{\partial\varphi_j}{\partial\rho} 
		- \Lambda^e\dfrac{\partial q_j}{\partial\rho}\right] 
		+ \dfrac{\partial\rho_{,k}}{\partial x_j} \left[\dfrac{\partial\varphi_j}{\partial\rho_{,k}} 
		- \Lambda^e\dfrac{\partial q_j}{\partial\rho_{,k}} \right] 
		%\nonumber \\
		+ \dfrac{\partial\theta}{\partial x _j} \left[ \dfrac{\partial\varphi_j}{\partial\theta} 
		- \Lambda^e \dfrac{\partial q_j}{\partial\theta} \right]  
		\\
		& \quad + D_{ij} \left[ -\rho^2\delta_{ij} \left(\overline{\dfrac{\partial s}{\partial\rho}} 
		- \Lambda^e \dfrac{\partial e}{\partial\rho}\right) + \Lambda^et_{ij} 
		- \Lambda^{\nabla\rho}_k\dfrac{\partial\rho}{\partial x_k}\delta_{ij} 
		- \Lambda^{\nabla\rho}_j\dfrac{\partial\rho}{\partial x_i} \right] = \Sigma \geq 0. 
		\nonumber 
	\end{align}
\end{lemma}

\begin{proof}
	Inequality \eqref{rezidualna nejednakost 2 Kortevegovi fluidi} follows from \eqref{osnovna rezidualna nejednakost Kortevegovi fluidi} after transformation of the coefficients. In particular, the coefficient of $\partial\rho/\partial x_j$ in \eqref{osnovna rezidualna nejednakost Kortevegovi fluidi} is transformed as
	\begin{align} \label{koeficijent uz rho po x_j Kortevegovi fluidi}
		& \left(s+\rho\overline{\dfrac{\partial s}{\partial\rho}} - \Lambda^\rho - \Lambda^{v_i}v_i 
		- \dfrac{1}{2}\Lambda^e|\mathbf{v}|^2-\Lambda^ e e 
		- \Lambda^e\rho\dfrac{\partial e}{\partial\rho}\right)v_j 
		+ \dfrac{\partial\varphi_j}{\partial\rho} 
		+ \left( \Lambda^{v_i} + \Lambda^ev_i \right) \dfrac{\partial t_{ij}}{\partial\rho} 
		- \Lambda^e\dfrac{\partial q_j}{\partial\rho} 
		\\
		& \quad \overset{\eqref{Trho Kortevegovi fluidi},\eqref{Tv_i' Kortevegovi fluidi}}{=}\dfrac{\partial\varphi_j}{\partial\rho}-\Lambda^e\dfrac{\partial q_j}{\partial\rho}. 
		\nonumber 
	\end{align}
	Also, the coefficient of $\partial\rho_{,k}/\partial x_j$ in \eqref{osnovna rezidualna nejednakost Kortevegovi fluidi} reads 
	\begin{align} \label{koeficijent uz rho_,k po x_j Kortevegovi fluidi}
		& \left(\rho\overline{\dfrac{\partial s}{\partial\rho_{,k}}} 
		- \Lambda^e\rho\dfrac{\partial e}{\partial\rho_{,k}} - \Lambda^{\nabla\rho}_k\right)v_j 
		+ \dfrac{\partial\varphi_j}{\partial\rho_{,k}} 
		+ \left( \Lambda^{v_i} + \Lambda^e v_i \right) \dfrac{\partial t_{ij}}{\partial\rho_{,k}} 
		- \Lambda^e \dfrac{\partial q_j}{\partial\rho_{,k}} 
		\\
		& \quad \overset{\eqref{Trho_,k Kortevegovi fluidi},\eqref{Tv_i' Kortevegovi fluidi}}{=} \dfrac{\partial\varphi_j}{\partial\rho_{,k}}-\Lambda^e\dfrac{\partial q_j}{\partial\rho_{,k}}. 
		\nonumber 
	\end{align}
	Similarly, the coefficient of $\partial\theta/\partial x_j$ in \eqref{osnovna rezidualna nejednakost Kortevegovi fluidi} is transformed to 
	\begin{align} \label{koeficijent uz theta po x_j Kortevegovi fluidi}
		& \left(\rho\overline{\dfrac{\partial s}{\partial\theta}} 
		- \Lambda^e\rho\dfrac{\partial e}{\partial\theta}\right) v_j 
		+ \dfrac{\partial\varphi_j}{\partial\theta} 
		+ \left( \Lambda^{v_i} + \Lambda^e v_i \right) \dfrac{\partial t_{ij}}{\partial\theta} 
		- \Lambda^e\dfrac{\partial q_j}{\partial\theta} 
		\\
		& \quad\overset{\eqref{Ttheta Kortevegovi fluidi},\eqref{Tv_i' Kortevegovi fluidi}}{=}\dfrac{\partial\varphi_j}{\partial\theta}-\Lambda^e\dfrac{\partial q_j}{\partial\theta}. 
		\nonumber 
	\end{align}
	Finally, coefficient of $D_{ij}$ in \eqref{osnovna rezidualna nejednakost Kortevegovi fluidi} becomes 
	\begin{align} \label{koeficijent uz D_ij Kortevegovi fluidi}
		& \rho \left(s - \Lambda^\rho + \dfrac{1}{2}\Lambda^e\vert\mathbf{v}\vert^2 
		- \Lambda^e e \right) \delta_{ij} 
		- \rho \left( \Lambda^{v_i} + \Lambda^e v_i \right) v_j
		+ \Lambda^et_{ij} - \Lambda^{\nabla\rho}_k\dfrac{\partial\rho}{\partial x_k}\delta_{ij} 
		- \Lambda^{\nabla\rho}_j \dfrac{\partial\rho}{\partial x_i} \\
		& \quad \overset{\eqref{modifikovana Trho Kortevegovi fluidi},\eqref{Tv_i' Kortevegovi fluidi}}{=} 
		- \rho^2 \left(\overline{\dfrac{\partial s}{\partial\rho}} 
		- \Lambda^e\dfrac{\partial e}{\partial\rho}\right) \delta_{ij} 
		+ \Lambda^e t_{ij} - \Lambda^{\nabla\rho}_k \dfrac{\partial\rho}{\partial x_k}\delta_{ij} 
		- \Lambda^{\nabla\rho}_j \dfrac{\partial\rho}{\partial x_i}. 
		\nonumber 
	\end{align}
	Using \eqref{koeficijent uz rho po x_j Kortevegovi fluidi}, \eqref{koeficijent uz rho_,k po x_j Kortevegovi fluidi}, \eqref{koeficijent uz theta po x_j Kortevegovi fluidi} and \eqref{koeficijent uz D_ij Kortevegovi fluidi}, the residual inequality \eqref{osnovna rezidualna nejednakost Kortevegovi fluidi} is reduced to \eqref{rezidualna nejednakost 2 Kortevegovi fluidi}. 
\end{proof}

Equation \eqref{rezidualna nejednakost 2 Kortevegovi fluidi} presents the simplest form of the residual inequality that can be reached using the consequences of equations \eqref{Trho Kortevegovi fluidi}--\eqref{TD_kl Kortevegovi fluidi} and \eqref{X_jrho_,kl Kortevegovi fluidi}--\eqref{X_jD_kl Kortevegovi fluidi}. However, it cannot be used for derivation of the admissible constitutive relations. 

Apart from non-convective fluxes $t_{ij}$ and $q_{j}$, \eqref{rezidualna nejednakost 2 Kortevegovi fluidi} inherits the entropy flux $\varphi_{j}$ as well. This raises a question about relation between the heat flux and the entropy flux---a delicate problem for which relevant studies do not offer a unique solution. For example, in the case of viscous and heat conducting fluids \cite{liu2002continuum} showed that these vector fields are parallel. \cite{anderson1998diffuse} deduce that, apart from temperature gradient, heat flux contains an extra term proportional to the product of density gradient and divergence of velocity, which does not appear in the entropy flux. \cite{van2023holographic} imposes a relation between the heat and entropy fluxes as a constraint. Entropy production in the form \eqref{rezidualna nejednakost 2 Kortevegovi fluidi} does not provide a proper basis for the answer to this question. 

Another reason for transformation of \eqref{rezidualna nejednakost 2 Kortevegovi fluidi} is related to the equilibrium conditions. If the equilibrium state is characterized by vanishing of the entropy production, then vanishing of the temperature and velocity gradient is sufficient condition for equilibrium of viscous and heat conducting fluids (see \cite{liu2002continuum}). In the case of Korteweg fluids, the situation is more complicated. Entropy production contains also the gradients of density, which are not supposed to vanish in equilibrium. Again, we cannot base the analysis upon \eqref{rezidualna nejednakost 2 Kortevegovi fluidi}. 

Therefore, our aim is to transform the residual inequality \eqref{rezidualna nejednakost 2 Kortevegovi fluidi} in such a way that entropy flux can be deduced from it independently of the stress tensor and heat flux, and that equilibrium conditions can be fulfilled. This final form of entropy production and residual inequality will resemble the one proposed by \cite{anderson1998diffuse}, in equation (15). It will contain an extra additive term of divergence form, which will facilitate the complete analysis of constitutive relations. 

\begin{theorem}
	The final form of the residual inequality reads 
	\begin{align} \label{rezidualna nejednakost final Kortevegovi fluidi}
		& \dfrac{\partial}{\partial x_j} \left( \varphi_j - \Lambda^eq_j 
		- \Lambda^{\nabla\rho}_j \rho D_{kl}\delta_{kl}\right) 
		+ \dfrac{\partial\Lambda^e}{\partial\theta} \dfrac{\partial\theta}{\partial x_j}q_j 
		\\
		& \quad + D_{ij} \left[ \left( -\rho^2 \left( \overline{\dfrac{\partial s}{\partial\rho}} 
		- \Lambda^e \dfrac{\partial e}{\partial\rho} \right) 
		+ \dfrac{\partial\Lambda^{\nabla\rho}_k}{\partial x_k}\rho\right) \delta_{ij} 
		+ \Lambda^et_{ij} - \Lambda^{\nabla\rho}_j\dfrac{\partial\rho}{\partial x_i} \right]
		= \Sigma \geq 0. 
		\nonumber 
	\end{align}
\end{theorem}

\begin{proof}
	We shall start by transforming the first two terms in the residual inequality \eqref{rezidualna nejednakost 2 Kortevegovi fluidi}. Taking into account that entropy flux $\varphi_{j}$ and heat flux $q_{j}$ are determined by the constitutive functions
	\begin{align*}
		\varphi_j & = \varphi_{j}(\rho, \rho_{,k}, \rho_{,kl}, \theta, \theta_{,k}, D_{kl}), 
		\\
		q_j & = q_{j}(\rho, \rho_{,k}, \rho_{,kl}, \theta, \theta_{,k}, D_{kl}), \quad 
		j = 1, 2, 3. 
	\end{align*} 
	It follows 
	\begin{align*}
		\dfrac{\partial\varphi_j}{\partial x_j}
		& = \dfrac{\partial\varphi_j}{\partial\rho}\dfrac{\partial\rho}{\partial x_j} 
		+ \dfrac{\partial\varphi_j}{\partial\rho_{,k}}\dfrac{\partial\rho_{,k}}{\partial x_j} 
		\\
		& \quad + \dfrac{\partial\varphi_j}{\partial\rho_{,kl}}\dfrac{\partial\rho_{,kl}}{\partial x_j}+\dfrac{\partial\varphi_j}{\partial\theta}\dfrac{\partial\theta}{\partial x_j}+\dfrac{\partial\varphi_j}{\partial\theta_{,k}}\dfrac{\partial\theta_{,k}}{\partial x_j}+\dfrac{\partial\varphi_j}{\partial D_{kl}}\dfrac{\partial D_{kl}}{\partial x_j},
	\end{align*}
	and analogously for the heat flux. Taking this into account, first line in the residual inequality \eqref{rezidualna nejednakost 2 Kortevegovi fluidi} in Lemma \ref{Lema:Rezidualna nejednakost 2} can be transformed in the following way
	\begin{align*}
		& \dfrac{\partial\rho}{\partial x_j} \left[ \dfrac{\partial\varphi_j}{\partial\rho} 
		- \Lambda^e\dfrac{\partial q_j}{\partial\rho} \right] 
		+ \dfrac{\partial\rho_{,k}}{x_j} \left[ \dfrac{\partial\varphi_j}{\partial\rho_{,k}} 
		- \Lambda^e\dfrac{\partial q_j}{\partial\rho_{,k}} \right] \\
		& \quad = \dfrac{\partial\varphi_j}{\partial\rho} \dfrac{\partial\rho}{\partial x_j} 
		+ \dfrac{\partial\varphi_j}{\partial\rho_{,k}} \dfrac{\partial\rho_{,k}}{\partial x_j} 
		- \Lambda^e\left[\dfrac{\partial q_j}{\partial\rho}\dfrac{\partial\rho}{\partial x_j} 
		+ \dfrac{\partial q_j}{\partial\rho_{,k}} \dfrac{\partial\rho_{,k}}{\partial x_j}\right] \\
		& \quad = \dfrac{\partial\varphi_j}{\partial x_j} - \left(   
		\dfrac{\partial\varphi_j}{\partial\rho_{,kl}} \dfrac{\partial\rho_{,kl}}{\partial x_j} 
		+ \dfrac{\partial\varphi_j}{\partial\theta} \dfrac{\partial\theta}{\partial x_j} 
		+ \dfrac{\partial\varphi_j}{\partial\theta_{,k}} \dfrac{\partial\theta_{,k}}{\partial x_j} 
		+ \dfrac{\partial\varphi_j}{\partial D_{kl}} \dfrac{\partial D_{kl}}{\partial x_j} \right) \\
		& \qquad - \Lambda^e \left[ \dfrac{\partial q_j}{\partial x_j} 
		- \left( \dfrac{\partial q_j}{\partial\rho_{,kl}} \dfrac{\partial\rho_{,kl}}{\partial x_j} 
		+ \dfrac{\partial q_j}{\partial\theta}\dfrac{\partial\theta}{\partial x_j} 
		+ \dfrac{\partial q_j}{\partial\theta_{,k}} \dfrac{\partial\theta_{,k}}{\partial x_j} 
		+ \dfrac{\partial q_j}{\partial D_{kl}} \dfrac{\partial D_{kl}}{\partial x_j}\right) \right] \\
		& \quad = \dfrac{\partial\varphi_j}{\partial x_j} 
		- \Lambda^e \dfrac{\partial q_j}{\partial x_j} % \\ & \qquad 
		- \left[ \left( \dfrac{\partial\varphi_j}{\partial{\rho_{,kl}}} \dfrac{\partial\rho_{,kl}}{\partial x_j} 
		- \Lambda^e \dfrac{\partial q_j}{\partial{\rho_{,kl}}} \dfrac{\partial\rho_{,kl}}{\partial x_j} \right) 
		+ \left( \dfrac{\partial\varphi_j}{\partial{\theta}}\dfrac{\partial\theta}{\partial x_j} 
		- \Lambda^e\dfrac{\partial q_j}{\partial{\theta}} \dfrac{\partial\theta}{\partial x_j} \right) \right. \\
		& \qquad \left. + \left( \dfrac{\partial\varphi_j}{\partial{\theta_{,k}}} \dfrac{\partial\theta_{,k}}{\partial x_j} 
		- \Lambda^e \dfrac{\partial q_j}{\partial{\theta_{,k}}} 
		\dfrac{\partial\theta_{,k}}{\partial x_j}\right) 
		+ \left(\dfrac{\partial\varphi_j}{\partial D_{kl}}\dfrac{\partial D_{kl}}{\partial x_j} 
		- \Lambda^e \dfrac{\partial q_j}{\partial D_{kl}} 
		\dfrac{\partial D_{kl}}{\partial x_j} \right) \right]\\
		& \overset{\eqref{X_jrho_,kl prim Kortevegovi fluidi}, \eqref{X_jtheta_,k prim Kortevegovi fluidi}, \eqref{X_jD_kl prim Kortevegovi fluidi}} {=} 
		\dfrac{\partial\varphi_j}{\partial x_j} - \Lambda^e \dfrac{\partial q_j}{\partial x_j} 
		- \left( \dfrac{\partial\varphi_j}{\partial{\theta}}\dfrac{\partial\theta}{\partial x_j} 
		- \Lambda^e\dfrac{\partial q_j}{\partial{\theta}}\dfrac{\partial\theta}{\partial x_j} \right) 
		- \Lambda^{\nabla\rho}_j\rho\delta_{kl} \dfrac{\partial D_{kl}}{\partial x_j}.
	\end{align*}
	Using this transformation, inequality \eqref{rezidualna nejednakost 2 Kortevegovi fluidi} can be written in the form 
	%\begin{equation}
	\begin{align} \label{rezidualna nejednakost 3 Kortevegovi fluidi}
		& \dfrac{\partial\varphi_j}{\partial x_j} - \Lambda^e \dfrac{\partial q_j}{\partial x_j} - \Lambda^{\nabla\rho}_j\rho\delta_{kl}\dfrac{\partial D_{kl}}{\partial x_j}
		\\
		& \quad +D_{ij}\left[-\rho^2\delta_{ij}\left(\overline{\dfrac{\partial s}{\partial\rho}}-\Lambda^e\dfrac{\partial e}{\partial\rho}\right)+\Lambda^et_{ij}-\Lambda^{\nabla\rho}_k\dfrac{\partial\rho}{\partial x_k}\delta_{ij}-\Lambda^{\nabla\rho}_j\dfrac{\partial\rho}{\partial x_i}\right]=\Sigma\geq 0. 
		\nonumber 
	\end{align}
	%\end{equation}
	Having in mind that $\Lambda^e$ depends only on temperature $\theta$ (see equation \eqref{mnozitelj Lambda-e Kortevegovi fluidi}), residual inequality \eqref{rezidualna nejednakost 3 Kortevegovi fluidi} can be written as
	\begin{align*}
		& \dfrac{\partial}{\partial x_j}\varphi_j - \Lambda^e\dfrac{\partial}{\partial x_j}q_j 
			- \Lambda^{\nabla\rho}_j\rho\delta_{kl} \dfrac{\partial D_{kl}}{\partial x_j} 
			+ D_{ij} \left[ -\rho^2\delta_{ij}\left(\overline{\dfrac{\partial s}{\partial\rho}} 
			- \Lambda^e\dfrac{\partial e}{\partial\rho} \right) + \Lambda^et_{ij} 
			- \Lambda^{\nabla\rho}_k\dfrac{\partial\rho}{\partial x_k}\delta_{ij} 
			- \Lambda^{\nabla\rho}_j\dfrac{\partial\rho}{\partial x_i} \right] 
		\\
		& \quad = \dfrac{\partial}{\partial x_j} \left( \varphi_j - \Lambda^eq_j\right) 
			+ \dfrac{\partial\Lambda^e}{\partial x_j}q_j 
			- \Lambda^{\nabla\rho}_j\rho\delta_{kl} \dfrac{\partial D_{kl}}{\partial x_j} 
		\\ 
		& \qquad + D_{ij} \left[ \left( -\rho^2\left(\overline{\dfrac{\partial s}{\partial\rho}} 
			- \Lambda^e \dfrac{\partial e}{\partial\rho} \right) 
			- \Lambda^{\nabla\rho}_k\dfrac{\partial\rho}{\partial x_k} \right) \delta_{ij} 
			+ \Lambda^e t_{ij} - \Lambda^{\nabla\rho}_j\dfrac{\partial\rho}{\partial x_i} \right] \\
			& \quad = \dfrac{\partial}{\partial x_j} \left( \varphi_j - \Lambda^eq_j \right) 
			+ \dfrac{\partial\Lambda^e}{\partial x_j}q_j 
			- \dfrac{\partial}{\partial x_j} \left( \Lambda^{\nabla\rho}_j\rho D_{kl} \delta_{kl} \right) 
			+ \dfrac{\partial}{\partial x_j}\left(\Lambda^{\nabla\rho}_j\rho\right) \delta_{kl} D_{kl} \\
			& \qquad + D_{ij} \left[ \left( -\rho^2\left(\overline{\dfrac{\partial s}{\partial\rho}} 
			- \Lambda^e\dfrac{\partial e}{\partial\rho} \right) 
			- \Lambda^{\nabla\rho}_k\dfrac{\partial\rho}{\partial x_k} \right) \delta_{ij} 
			+ \Lambda^e t_{ij} - \Lambda^{\nabla\rho}_j\dfrac{\partial\rho}{\partial x_i} \right] \\
			& \quad = \dfrac{\partial}{\partial x_j} \left( \varphi_j - \Lambda^e q_j 
			- \Lambda^{\nabla\rho}_j \rho D_{kl}\delta_{kl} \right) 
			+ \dfrac{\partial\Lambda^e}{\partial\theta} \dfrac{\partial\theta}{\partial x_j} q_j \\
			& \qquad + D_{ij} \left[ \left( -\rho^2\left(\overline{\dfrac{\partial s}{\partial\rho}} 
			- \Lambda^e \dfrac{\partial e}{\partial\rho} \right) 
			- \Lambda^{\nabla\rho}_k \dfrac{\partial\rho}{\partial x_k} 
			+ \dfrac{\partial}{\partial x_k} \left( \Lambda^{\nabla\rho}_k\rho \right) \right) \delta_{ij} 
			% \right. \\ & \qquad \left. 
			+ \Lambda^e t_{ij} 
			- \Lambda^{\nabla\rho}_j \dfrac{\partial\rho}{\partial x_i} \right] \quad = \Sigma \geq 0.
	\end{align*}
	Since $\dfrac{\partial}{\partial x_k}\left(\Lambda^{\nabla\rho}_k\rho\right)=\dfrac{\partial\Lambda^{\nabla\rho}_k}{\partial x_k}\rho+\Lambda^{\nabla\rho}_k\dfrac{\partial\rho}{\partial x_k}$, the last inequality becomes 
	\begin{align*}
		& \dfrac{\partial}{\partial x_j} \left( \varphi_j - \Lambda^eq_j 
			- \Lambda^{\nabla\rho}_j\rho D_{kl}\delta_{kl} \right) 
			+ \dfrac{\partial\Lambda^e}{\partial\theta} \dfrac{\partial\theta}{\partial x_j} q_j 
		\\
		& \quad + D_{ij} \left[ \left( -\rho^2 \left( \overline{\dfrac{\partial s}{\partial\rho}} 
			- \Lambda^e \dfrac{\partial e}{\partial\rho}\right) + \dfrac{\partial\Lambda^{\nabla\rho}_k}{\partial x_k}\rho \right) \delta_{ij} 
			+ \Lambda^e t_{ij} - \Lambda^{\nabla\rho}_j \dfrac{\partial\rho}{\partial x_i} \right] 
			= \Sigma \geq 0,
	\end{align*}
	which proves \eqref{rezidualna nejednakost final Kortevegovi fluidi}.
\end{proof}

\subsection{Constitutive relations} 

Starting point in derivation of the admissible constitutive relations is the residual inequality \eqref{rezidualna nejednakost final Kortevegovi fluidi}. It contains the divergence of the vector field and terms which can be interpreted as products of thermodynamic forces, $\partial \theta/\partial x_{j}$ and $D_{ij}$, and thermodynamic fluxes $q_{j}$ and $t_{ij}$. By the proper choice of thermodynamic fluxes, one can secure the non-negative sign of the products. However, the sign of divergence of the vector field cannot be controlled, and thus have to be eliminated (see \cite{maugin1999thermomechanics} for discussion on this issue). For this reason, constitutive relations will be determined in three steps: 
\begin{enumerate}
	\item elimination of divergence term, 
	\item derivation of thermodynamic fluxes in equilibrium, and
	\item derivation of thermodynamic fluxes for arbitrary thermodynamic process, i.e. in non-equilibrium. 
\end{enumerate}
First step is given in the next Proposition. 

\begin{proposition} \label{Prop:Protok entropije Kortevegovi fluidi}
	Entropy flux $\varphi_{j}$ is determined by the following relation
	\begin{equation}\label{protok entropije Kortevegovi fluidi}
		\varphi_{j} = \dfrac{1}{\theta} \left( q_j + \rho \alpha(\rho,\theta) 
		\dfrac{\partial\rho}{\partial x_j} \dfrac{\partial v_k}{\partial x_k} \right). 
	\end{equation}
\end{proposition} 

\begin{proof}
	Since divergence of the vector field in equation \eqref{rezidualna nejednakost final Kortevegovi fluidi} cannot be of definite sign, it will be assumed that vector field is equal to zero vector field
	\begin{equation*}
		\varphi_j - \Lambda^e q_j - \Lambda^{\nabla\rho}_j\rho D_{kl}\delta_{kl} = 0_{j}. 
	\end{equation*}
	Taking into account that multipliers $\Lambda^{e}$ and $\Lambda^{\nabla\rho}_j$ are determined by \eqref{mnozitelj Lambda-e Kortevegovi fluidi} and \eqref{mnozitelj Lambda-grad-rho Kortevegovi fluidi}, respectively, and that $D_{kl} \delta_{kl} = \partial v_k/\partial x_k$, one arrives at \eqref{protok entropije Kortevegovi fluidi}. 
\end{proof} 

Note that \eqref{protok entropije Kortevegovi fluidi} does not represent the most general form of the entropy flux. It may contain an additive divergence-free term, but there is neither mathematical method, nor physical argument which can explicitly determine its structure. Finally, using \eqref{protok entropije Kortevegovi fluidi}, the entropy production \eqref{rezidualna nejednakost final Kortevegovi fluidi} can be written as 
\begin{equation} \label{produkcija entropije Sigma Kortevegovi fluidi}
	\Sigma = \dfrac{\partial\Lambda^e}{\partial\theta} \dfrac{\partial\theta}{\partial x_j} q_j 
	% \\ & \quad 
	+ D_{ij} \left[ \left( -\rho^2 \left( \overline{\dfrac{\partial s}{\partial\rho}} 
	- \Lambda^e \dfrac{\partial e}{\partial\rho}\right) 
	+ \dfrac{\partial\Lambda^{\nabla\rho}_k}{\partial x_k}\rho \right) \delta_{ij} 
	+ \Lambda^e t_{ij} - \Lambda^{\nabla\rho}_j \dfrac{\partial\rho}{\partial x_i} \right]. 
\end{equation}

In the second step it is required to define the equilibrium state as a state in which the entropy production \eqref{produkcija entropije Sigma Kortevegovi fluidi} vanishes, and attains minimum at the same time
\begin{equation} \label{ravnoteza Sigma Kortevegovi fluidi}
	\Sigma\big|_{eq} = 0, \quad \Sigma\big|_{eq} = \min_{\mathrm{state space}} \Sigma. 
\end{equation}
To do this in a proper way, entropy production \eqref{produkcija entropije Sigma Kortevegovi fluidi} should be expressed in such a way that there are no ``hidden'' terms. Next Lemma provides an appropriate form of the entropy production. 
\begin{lemma} \label{Lema:Produkcija entropije-konacna}
	Entropy production $\Sigma$ of Korteweg fluids has the following form
	\begin{align} \label{produkcija entropije Sigma Kortevegovi fluidi_konacna}
		& \Sigma = - \dfrac{1}{\theta^{2}} \dfrac{\partial\theta}{\partial x_j} q_j 
		- \dfrac{\rho}{\theta^{2}} \left[ \alpha(\rho, \theta) 
		- \theta \dfrac{\partial \alpha(\rho, \theta)}{\partial \theta} \right] 
		\dfrac{\partial \rho}{\partial x_k} \dfrac{\partial \theta}{\partial x_k} 
		D_{ij} \delta_{ij}
		\\ 
		& \quad + \dfrac{1}{\theta} D_{ij} \left\{ \left[ p 
		+ \rho \dfrac{\partial \alpha(\rho, \theta)}{\partial \rho}  
		\dfrac{\partial \rho}{\partial x_k} \dfrac{\partial \rho}{\partial x_k} 
		+ \rho \alpha(\rho, \theta) \dfrac{\partial^{2} \rho}{\partial x_k \partial x_k} \right] 
		\delta_{ij} + t_{ij} - \alpha(\rho, \theta) 
		\dfrac{\partial \rho}{\partial x_i} \dfrac{\partial \rho}{\partial x_j} \right\}. 
		\nonumber 
	\end{align}
\end{lemma} 

\begin{proof}
	Multipliers $\Lambda^e$ i $\Lambda^{\nabla\rho}_i$ of Korteweg fluids are determined by \eqref{mnozitelj Lambda-e Kortevegovi fluidi} and \eqref{mnozitelj Lambda-grad-rho Kortevegovi fluidi}, respectively:
	\begin{equation*}
		\Lambda^e = \dfrac{1}{\theta}, \quad 
		\Lambda^{\nabla\rho}_i = \dfrac{1}{\theta} \alpha(\rho, \theta) 
		\dfrac{\partial \rho}{\partial x_i}. 
	\end{equation*}
	This implies  
	\begin{align} \label{Lambda-grad-rho po xj Kortevegovi fluidi}
		\dfrac{\partial\Lambda^{\nabla\rho}_k}{\partial x_k} & = 
		\dfrac{\partial}{\partial x_k} \left( \dfrac{1}{\theta} \alpha(\rho,\theta) 
		\dfrac{\partial\rho}{\partial x_k} \right) 
		\\
		& = \left[ - \dfrac{1}{\theta^{2}} \alpha(\rho,\theta) 
		+ \dfrac{1}{\theta} \dfrac{\partial \alpha(\rho,\theta)}{\partial\theta} \right] \dfrac{\partial\theta}{\partial x_k} \dfrac{\partial\rho}{\partial x_k} 
		+ \dfrac{1}{\theta} \dfrac{\partial \alpha(\rho,\theta)}{\partial\rho} 
		\dfrac{\partial\rho}{\partial x_k} \dfrac{\partial\rho}{\partial x_k}
		+ \dfrac{1}{\theta} \alpha(\rho,\theta)\dfrac{\partial^2\rho}{\partial x_k\partial x_k}. 
		\nonumber 
	\end{align}
	Using the definition of total derivative \eqref{totalni izvod} we may also transform 
	%\begin{equation} 
	\begin{align} \label{deo izraza za RI Kortevegovi fluidi}
		- \rho^2 \left( \overline{\dfrac{\partial s}{\partial\rho}} 
		- \Lambda^e \dfrac{\partial e}{\partial\rho} \right) 
		& = - \rho^2 \left( \dfrac{\partial s}{\partial e} \dfrac{\partial e}{\partial\rho} 
		+ \dfrac{\partial s}{\partial\rho} - \Lambda^e \dfrac{\partial e}{\partial\rho} \right) 
		\nonumber \\
		& \overset{\eqref{Lambda_e Kortevegovi fluidi}}{=} 
		- \rho^2 \dfrac{\partial s}{\partial\rho} 
		\\
		& \overset{\eqref{konstitutivne relacije s hipoteza 2}_{2}}{=} \dfrac{p}{\theta}. 
		\nonumber 
	\end{align}
	%\end{equation}
	Finally, inserting \eqref{mnozitelj Lambda-e Kortevegovi fluidi}, \eqref{Lambda-grad-rho po xj Kortevegovi fluidi} and \eqref{deo izraza za RI Kortevegovi fluidi} into \eqref{produkcija entropije Sigma Kortevegovi fluidi}, after straightforward transformations one obtains \eqref{produkcija entropije Sigma Kortevegovi fluidi_konacna}. 
\end{proof}

It is now possible to express the equilibrium conditions in terms of the elements of state space.

\begin{lemma} \label{Lema:Ravnoteza Kortevegovih fluida}
	Sufficient conditions of equilibrium of Korteweg fluids are 
	\begin{equation} \label{ravnoteza Kortevegovi fluidi}
		D_{ij}|_{eq} = 0, \quad \dfrac{\partial \theta}{\partial x_{j}}\bigg|_{eq} = 0, 
		\quad i,j = 1,2,3.  
	\end{equation}
\end{lemma}
\begin{proof}
	Since $\partial \rho/\partial x_k$ and $\partial^{2} \rho/\partial x_k \partial x_k$ do not have to vanish in the equilibrium state of Korteweg fluids, equilibrium condition \eqref{ravnoteza Sigma Kortevegovi fluidi}$_{1}$ will be identically satisfied if \eqref{ravnoteza Kortevegovi fluidi} holds. 
\end{proof}

Lemmas \ref{Lema:Produkcija entropije-konacna} and \ref{Lema:Ravnoteza Kortevegovih fluida} are the steps that were necessary to determine thermodynamic fluxes in equilibrium state. Final step consists in expressing the condition \eqref{ravnoteza Sigma Kortevegovi fluidi}$_{2}$ of local minimum of the entropy production in equilibrium through the necessary conditions  
\begin{equation} \label{stacionarnost Kortevegovi fluidi} 
	\dfrac{\partial \Sigma}{\partial D_{ij}} \bigg|_{eq} = 0, \quad 
	\dfrac{\partial \Sigma}{\partial \theta_{,j}} \bigg|_{eq} = 0.
\end{equation}

\begin{theorem} \label{Te:Ravnotezni protoci}
	Thermodynamic fluxes, stress tensor $t_{ij}$ and heat flux $q_{j}$, have the following form in equilibrium state 
	\begin{align} 
		\label{tenzor napona equilibrium Kortevegovi fluidi}
		t_{ij}^{eq} & = - \left[ p + \rho \dfrac{\partial \alpha(\rho,\theta)}{\partial\rho} 
		\dfrac{\partial\rho}{\partial x_k} \dfrac{\partial\rho}{\partial x_k} 
		+ \rho \alpha(\rho,\theta) \dfrac{\partial^2\rho}{\partial x_k\partial x_k} \right] \delta_{ij} 
		+ \alpha(\rho,\theta) \dfrac{\partial\rho}{\partial x_i} \dfrac{\partial\rho}{\partial x_j}, 
		\quad i,j,k = 1,2,3,
		\\
		\label{toplotni protok equilibrium Kortevegovi fluidi}
		q_{j}^{eq} & = 0_{j}, \quad j = 1,2,3.
	\end{align}
\end{theorem}

\begin{proof}
	Using \eqref{produkcija entropije Sigma Kortevegovi fluidi_konacna}, from the stationarity condition \eqref{stacionarnost Kortevegovi fluidi}$_{1}$ it follows 
	\begin{equation*}
		\dfrac{1}{\theta} \left\{ 
			\left[ p + \rho \dfrac{\partial \alpha(\rho, \theta)}{\partial \rho} 
			\dfrac{\partial \rho}{\partial x_k} \dfrac{\partial \rho}{\partial x_k} 
			+ \rho \alpha(\rho, \theta) \dfrac{\partial^{2} \rho}{\partial x_k \partial x_k} \right]
			\delta_{ij} 
			+ t_{ij}^{eq} 
			- \alpha(\rho, \theta) \dfrac{\partial \rho}{\partial x_i} \dfrac{\partial \rho}{\partial x_j} 
			\right\} = 0,
	\end{equation*} 
	which directly implies \eqref{tenzor napona equilibrium Kortevegovi fluidi}. On the other hand, from \eqref{stacionarnost Kortevegovi fluidi}$_{2}$ it follows 
	\begin{equation*}
		- \dfrac{1}{\theta^{2}} q_{j}^{eq} = 0, \quad j = 1,2,3, 
	\end{equation*}
	which leads directly to \eqref{toplotni protok equilibrium Kortevegovi fluidi}. 
\end{proof} 

In the third step---derivation of the non-equilibrium thermodynamic fluxes---we shall introduce non-equilibrium fluxes $t_{ij}^{\ast}$ and $q_{j}^{\ast}$ in the following way
\begin{equation*}
	t_{ij} = t_{ij}^{eq} + t_{ij}^{\ast}, \quad 
	q_{j} = q_{j}^{eq} + q_{j}^{\ast}. 
\end{equation*}
Inserting these expressions in equation \eqref{produkcija entropije Sigma Kortevegovi fluidi_konacna} and using the statement of Theorem \ref{Te:Ravnotezni protoci}, i.e. equations \eqref{tenzor napona equilibrium Kortevegovi fluidi} and \eqref{toplotni protok equilibrium Kortevegovi fluidi}, entropy production can be written as 
\begin{equation} \label{Produkcija entropije neravnotezna Kortevegovi fluidi}
	\Sigma = - \dfrac{1}{\theta^{2}} q_{k}^{\ast} \dfrac{\partial \theta}{\partial x_k} 
	- \dfrac{\rho}{\theta^{2}} \left[ \alpha(\rho, \theta) 
	- \theta \dfrac{\partial \alpha(\rho, \theta)}{\partial \theta} \right] 
	\dfrac{\partial \rho}{\partial x_k} \dfrac{\partial \theta}{\partial x_k} D_{ij} \delta_{ij}
	+ \dfrac{1}{\theta} t_{ij}^{\ast} D_{ij}. 
\end{equation} 
Non-equilibrium fluxes $t_{ij}^{\ast}$ and $q_{j}^{\ast}$ should be chosen in such a way that entropy production $\Sigma$ remains non-negative for any thermodynamic process 
\begin{equation} \label{Entropijska nejednakost Kortevegovi fluidi}
	\Sigma \geq 0. 
\end{equation} 
The simplest form of the constitutive relations is given by the next Theorem. 

\begin{theorem} \label{Te:Neravnotezni Protoci}
	Non-equilibrium fluxes $t_{ij}^{\ast}$ and $q_{j}^{\ast}$ which ensure the entropy inequality \eqref{Entropijska nejednakost Kortevegovi fluidi} can be determined by the following constitutive relations 
	\begin{itemize}
		\item Case 1
		\begin{align} \label{Konstitutivne relacije Kortevegovi fluidi-1}
			t_{ij}^{\ast} & = \nu \dfrac{\partial v_k}{\partial x_k} \delta_{ij} + 2 \mu D_{ij}^{D}, 
			\\ 
			q_{k}^{\ast} & = - \kappa \dfrac{\partial \theta}{\partial x_k} 
			- \rho \left[ \alpha(\rho, \theta) 
			- \theta \dfrac{\partial \alpha(\rho, \theta)}{\partial \theta}  \right] 
			\dfrac{\partial v_j}{\partial x_j} \dfrac{\partial \rho}{\partial x_k}, 
			\nonumber 
		\end{align}
		\item Case 2
		\begin{align} \label{Konstitutivne relacije Kortevegovi fluidi-2}
			t_{ij}^{\ast} & = \nu \dfrac{\partial v_k}{\partial x_k} \delta_{ij} + 2 \mu D_{ij}^{D} 
			+ \dfrac{\rho}{\theta} \left[ \alpha(\rho, \theta) 
			- \theta \dfrac{\partial \alpha(\rho, \theta)}{\partial \theta} \right] 
			\dfrac{\partial \rho}{\partial x_k} \dfrac{\partial \theta}{\partial x_k} \delta_{ij}, 
			\\ 
			q_{k}^{\ast} & = - \kappa \dfrac{\partial \theta}{\partial x_k}, 
			\nonumber 
		\end{align}
	\end{itemize}
	where $\nu, \mu, \kappa \geq 0$, and $D_{ij}^{D}$ is the deviatoric part of $D_{ij}$ 
	\begin{equation} \label{Dij deviatoric}
		D_{ij}^{D} = D_{ij} - \frac{1}{3} D_{kk} \delta_{ij}.
	\end{equation}. 
\end{theorem} 

\begin{proof}
	Let us write the entropy production \eqref{Produkcija entropije neravnotezna Kortevegovi fluidi} in the form 
	\begin{equation*} 
		\Sigma = - \dfrac{1}{\theta^{2}} \left\{ q_{k}^{\ast} + \rho \left[ \alpha(\rho, \theta) 
		- \theta \dfrac{\partial \alpha(\rho, \theta)}{\partial \theta} \right] 
		\dfrac{\partial \rho}{\partial x_k} D_{ij} \delta_{ij} \right\} \dfrac{\partial \theta}{\partial x_k} 
		+ \dfrac{1}{\theta} t_{ij}^{\ast} \left( \frac{1}{3} D_{kk} \delta_{ij} + D_{ij}^{D} \right), 
	\end{equation*}
	where $D_{ij}$ is represented as the sum of spherical and deviatoric part of the tensor 
	\begin{equation*}
		D_{ij} = \frac{1}{3} D_{kk} \delta_{ij} + D_{ij}^{D}, \quad 
		D_{kk} = \dfrac{\partial v_k}{\partial x_k}. 
	\end{equation*} 
	Choosing the constitutive relations as \eqref{Konstitutivne relacije Kortevegovi fluidi-1}, inserting them in the entropy production $\Sigma$, and taking into account that $D_{ij}^{D} \delta_{ij} = 0$, one arrives at 
	\begin{equation} \label{Entropijska nejednakost-konacna}
		\Sigma = \dfrac{\kappa}{\theta^{2}} 
		\dfrac{\partial \theta}{\partial x_k} \dfrac{\partial \theta}{\partial x_k} 
		+ \dfrac{\nu}{\theta} \left( \dfrac{\partial v_k}{\partial x_k} \right)^{2}
		+ \dfrac{2 \mu}{\theta} D_{ij}^{D} D_{ij}^{D}, 
	\end{equation}
	which along with $\nu, \mu, \kappa \geq 0$ ensures the entropy inequality \eqref{Entropijska nejednakost Kortevegovi fluidi}. On the other hand, if the entropy production \eqref{Produkcija entropije neravnotezna Kortevegovi fluidi} is written as 
	\begin{equation*} 
		\Sigma = - \dfrac{1}{\theta^{2}} q_{k}^{\ast} \dfrac{\partial \theta}{\partial x_k} 
		+ \dfrac{1}{\theta} \left\{ 
		- \dfrac{\rho}{\theta} \left[ \alpha(\rho, \theta) 
		- \theta \dfrac{\partial \alpha(\rho, \theta)}{\partial \theta} \right] 
		\dfrac{\partial \rho}{\partial x_k} \dfrac{\partial \theta}{\partial x_k} \delta_{ij}
		+ t_{ij}^{\ast} \right\} D_{ij},  
	\end{equation*} 
	then choosing the constitutive relations in the form \eqref{Konstitutivne relacije Kortevegovi fluidi-2} and inserting them into $\Sigma$ again leads to \eqref{Entropijska nejednakost-konacna}, which together with $\nu, \mu, \kappa \geq 0$ ensures the entropy inequality \eqref{Entropijska nejednakost Kortevegovi fluidi}.
\end{proof}

Cases 1 and 2 correspond to the entropic and energetic representations of the thermodynamic fluxes and forces, respectively, as analyzed by \cite{van2020variational} for Fourier--Navier--Stokes--Cahn--Hilliard--Korteweg fluids. The two representations are equivalent: they lead to the same entropy production \eqref{Entropijska nejednakost-konacna}, and the constitutive coefficients are related by a state-dependent transformation (\cite{van2020variational}). The cross-coupling term $\alpha - \theta\,\partial\alpha/\partial\theta$ reflects the temperature dependence of the capillary coefficient; when $\alpha$ is $\theta$-independent both cases collapse. Recent work \cite{bhattacharjee2025thermodynamically} extends this framework to nanoscale transport with gradient-dependent viscosity and thermal conductivity, naturally generalizing both cases. 

\subsection{On the structure of stress tensor} \label{Sec:Stress}

Equilibrium stress tensor \eqref{tenzor napona equilibrium Kortevegovi fluidi} obtained by Liu's method does not contain a $\nabla\otimes\nabla\rho$ (Hessian) term---it consists of an isotropic part (involving $\Delta\rho$ and $|\nabla\rho|^2$) and the dyadic product $\alpha\,\rho_{,i}\,\rho_{,j}$. However, there exist other well-known representations of the Korteweg stress tensor in the literature---most notably the Dunn--Serrin form (\cite{dunn1985interstitial}) and the one derived in \cite{szucs2025thermodynamic}---that include an explicit $\nabla\otimes\nabla\rho$ contribution. These alternative forms are \emph{divergence-equivalent} to \eqref{tenzor napona equilibrium Kortevegovi fluidi}: they differ from it by a term of the type $\partial_k Q^{kij}$, which does not contribute to the momentum balance. In other words, the bulk dynamics is insensitive to the choice of representative within this equivalence class (\cite{van2025remarks}). 

The origin of this non-uniqueness can be traced to a specific step in the Liu procedure. In the extended entropy inequality, the coefficient of $\nabla\otimes\nabla v$ --an element of the process direction space-- contains a contribution from the density gradient multiplier $\Lambda^{\nabla\rho}$. Since partial derivatives commute, $\nabla\otimes\nabla v$  is symmetric in $\nabla\otimes\nabla$; therefore only the symmetric part of its coefficient contributes to the Liu equations. If this symmetrization is carried out explicitly before reading off the entropy current density, as done in \cite{szucs2025thermodynamic}, the $\nabla\otimes\nabla\rho$ term appears in the stress tensor. In the present paper, the coefficient is not symmetrized, and the Hessian contribution is absent. Both choices are thermodynamically consistent---they lead to divergence-equivalent stress tensors and identical bulk dynamics (\cite{van2025remarks}). However, the more restrictive procedure (explicit symmetrization) exploits all available information from the process direction space and may therefore be regarded as the more complete application of Liu's method. The distinction becomes physically relevant at boundaries, where the traction $t_{ij}^{eq} n_j$ is directly measurable, and in dissipative extensions, where the entropy production $\Sigma = \frac{1}{\theta} t_{ij}^* D_{ij} + \ldots$ depends on the tensor itself, not merely its divergence. A systematic study of these consequences is left for future work. 

It is also worth noting that in the non-dissipative limit $\Sigma = 0$, the divergence of the equilibrium stress tensor \eqref{tenzor napona equilibrium Kortevegovi fluidi} can be cast in the \emph{holographic} form $\partial_j t_{ij}^{eq} = \rho\,\partial_i\Phi$, where $\Phi$ is a scalar potential expressible through the chemical potential and capillary contributions. This \emph{classical holographic property}---which reduces the momentum balance to a mass-point form $\dot{v}_i = -\partial_i\Phi$---is a general consequence of the Second Law for perfect continua, as demonstrated in \cite{szucs2025thermodynamic} and \cite{van2023holographic}. It connects the present results to the thermodynamic derivation of quantum-mechanical fluid equations.

\subsection{Gibbs relation for Korteweg fluids} \label{Sec:GibbsKorteweg}

The last step in the analysis of constitutive relations is derivation of the Gibbs relation for Korteweg fluids. It is a generalization of classical Gibbs relation since it includes differential of the mass density gradient. 

\begin{theorem} 
	Gibbs relation for Korteweg fluids reads: 
	\begin{equation} \label{Gibsova relacija Kortevegovi fluidi}
		\theta \mathrm{d}s = \mathrm{d}e - \dfrac{p}{\rho^2} \mathrm{d}\rho 
		+ \dfrac{\alpha(\rho,\theta)}{\rho} \rho_{,k} \mathrm{d}\rho_{,k}.
	\end{equation}
\end{theorem}

\begin{proof}
	Starting from equation \eqref{Trho_,k Kortevegovi fluidi}:
	\begin{equation*}
		\rho\overline{\dfrac{\partial s}{\partial\rho_{,k}}} 
		- \Lambda^e\rho\dfrac{\partial e}{\partial\rho_{,k}} - \Lambda^{\nabla\rho}_k = 0,
	\end{equation*}
	and using the total derivative \eqref{totalni izvod} we obtain:
	\begin{equation*}
		\rho\left(\left(\dfrac{\partial s}{\partial e} - \Lambda^e \right) 
		\dfrac{\partial e}{\partial\rho_{,k}} 
		+ \dfrac{\partial s}{\partial\rho_{,k}}\right) 
		- \Lambda_k^{\nabla\rho} = 0.
	\end{equation*}
	Taking into account the multipliers \eqref{Lambda_e Kortevegovi fluidi} and \eqref{mnozitelj Lambda-grad-rho Kortevegovi fluidi}, last equation reduces to:
	\begin{equation*}
		\rho\dfrac{\partial s}{\partial\rho_{,k}} = \Lambda_k^{\nabla\rho} 
		= \dfrac{1}{\theta} \alpha(\rho,\theta) \dfrac{\partial\rho}{\partial x_k},
	\end{equation*}
	or equivalently:
	\begin{equation} \label{partial s rho_,k}
		\dfrac{\partial s}{\partial\rho_{,k}} 
		= \dfrac{1}{\theta} \dfrac{\alpha(\rho,\theta)}{\rho} \dfrac{\partial\rho}{\partial x_k}.
	\end{equation}
	Using thermodynamic relations \eqref{konstitutivne relacije s hipoteza 2}, equations \eqref{s-argumenti Kortevegovi fluidi} and relation \eqref{partial s rho_,k}, total differential of the specific entropy may be written as:
	\begin{align*}
		\mathrm{d}s & \overset{\eqref{s-argumenti Kortevegovi fluidi}}{=} 
		\dfrac{\partial s}{\partial e} \mathrm{d}e 
		+ \dfrac{\partial s}{\partial\rho} \mathrm{d}\rho 
		+ \dfrac{\partial s}{\partial\rho_{,k}} \mathrm{d}\rho_{,k} 
		\\
		& \overset{\eqref{konstitutivne relacije s hipoteza 2},\eqref{partial s rho_,k}}{=} 
		\dfrac{1}{\theta} \mathrm{d}e - \dfrac{1}{\theta}\dfrac{p}{\rho^2} \mathrm{d}\rho 
		+ \dfrac{1}{\theta} \dfrac{\alpha(\rho,\theta)}{\rho} 
		\dfrac{\partial\rho}{\partial x_k} \mathrm{d}\rho_{,k} 
		\\
		& = \dfrac{1}{\theta} \left( \mathrm{d}e - \dfrac{p}{\rho^2} \mathrm{d}\rho 
		+ \dfrac{\alpha(\rho,\theta)}{\rho} \rho_{,k} \mathrm{d}\rho_{,k} \right), 
	\end{align*}
	which directly implies Gibbs relation for Korteweg fluids \eqref{Gibsova relacija Kortevegovi fluidi}. 
\end{proof}
Let us emphasize that \eqref{Gibsova relacija Kortevegovi fluidi} is a natural generalization of classical Gibbs relation to Korteweg fluids. In classical formulation, differential of specific entropy $s$ depends on differentials of equilibrium variables, specific internal energy $e$ and mass density $\rho$. In the case of Korteweg fluids, Gibbs relation includes differential of mass density gradient $\rho_{,k}$ which characterizes capillary stresses in local equilibrium. 

\subsection{On the structure of specific internal energy} \label{Sec:EnergyStructure}

Internal energy plays important role in the analysis of capillary effects. At first, specific energy $e$ is assumed in the form \eqref{konstitutivne relacije 1 opste}$_{1}$, but subsequent analysis did not impose any further restrictions regarding its structure. In other studies of Korteweg fluids, explicit dependence of the specific energy on mass density gradient was assumed (\cite{anderson1998diffuse}), or it was adapted to the structure of constitutive relations (\cite{heida2010compressible}). Here, we aim to reconsider our results from the perspective of these studies. 

In our notation, analysis in \cite{heida2010compressible} assumes specific internal energy in the form: 
\begin{equation} \label{specificna energija_Heida-Malek}
	e = e_{0}(\rho, \theta) + \hat{e}(\rho, \rho_{,k}), 
\end{equation}
On the other hand, stress tensor in local equilibrium reads:
\begin{equation} \label{tenzor napona_Heida-Malek}
	t_{ij}^{eq} = \left[ - p_{0} - \rho^{2} \frac{\partial \hat{e}}{\partial \rho} 
	+ \rho \frac{\partial}{\partial x_{k}} \left( 
	\rho \frac{\partial \hat{e}}{\partial \rho_{,k}} \right) \right] \delta_{ij}
	- \rho \frac{\partial \hat{e}}{\partial \rho_{,i}} \frac{\partial \rho}{\partial x_{j}}, 
\end{equation}
where $p_{0}$ is thermodynamic pressure: 
\begin{equation} \label{pritisak_Heida-Malek}
	p_{0} = \rho^{2} \frac{\partial e_{0}}{\partial \rho}. 
\end{equation}
Note that the structure of specific internal energy \eqref{specificna energija_Heida-Malek} is reflected on the structure of stress tensor in equilibrium \eqref{tenzor napona_Heida-Malek}, which is a consequence of the procedure of derivation of the constitutive relations. In the special case, when the following structure of specific energy is assumed: 
\begin{equation} \label{specificna energija kvadratna_Heida-Malek}
	e = e_{0}(\rho, \theta) + \frac{1}{2} A(\rho) \frac{\partial \rho}{\partial x_{k}} 
	\frac{\partial \rho}{\partial x_{k}}
\end{equation} 
stress tensor becomes:
\begin{equation} \label{tenzor napona 2_Heida-Malek}
	t_{ij}^{eq} = \left[ - p_{0} + \frac{\partial}{\partial \rho} 
	\left( \frac{1}{2} \rho^{2} A(\rho) \right) \frac{\partial \rho}{\partial x_{k}} 
	\frac{\partial \rho}{\partial x_{k}} 
	+ \rho^{2} A(\rho) \frac{\partial^{2} \rho}{\partial x_{k} \partial x_{k}} \right] \delta_{ij}
	- \rho A(\rho) \frac{\partial \rho}{\partial x_{i}} \frac{\partial \rho}{\partial x_{j}}. 
\end{equation}
When compared to stress tensor \eqref{tenzor napona equilibrium Kortevegovi fluidi} of our study, an equivalence may be observed with regard to density gradients. However, coefficients in \eqref{tenzor napona 2_Heida-Malek} are determined by the structure of specific internal energy. 

In our study, stress tensor \eqref{tenzor napona equilibrium Kortevegovi fluidi} and specific internal energy $e$ are independent. Structure of the specific energy is reflected on the stress tensor implicitly, through the pressure. If specific energy were assumed in the form \eqref{specificna energija_Heida-Malek}, then pressure $p$ would read: 
\begin{equation}
	p = \rho^{2} \frac{\partial e_{0}}{\partial \rho} 
	+ \rho^{2} \frac{\partial \hat{e}}{\partial \rho},
\end{equation}
where the first term corresponds to thermodynamic pressure $p_{0}$ from \eqref{pritisak_Heida-Malek}, whereas the second one is the Korteweg pressure. Furthermore, if particular form \eqref{specificna energija kvadratna_Heida-Malek} of the specific energy were assumed, then equilibrium stress tensor \eqref{tenzor napona equilibrium Kortevegovi fluidi} would become: 
\begin{equation*}
	t_{ij}^{eq} = - \left[ p_{0} + \left( \frac{1}{2} \rho^{2} A'(\rho) 
	+ \rho \dfrac{\partial \alpha(\rho,\theta)}{\partial\rho} \right)
	\dfrac{\partial\rho}{\partial x_k} \dfrac{\partial\rho}{\partial x_k} 
	+ \rho \alpha(\rho,\theta) \dfrac{\partial^2\rho}{\partial x_k\partial x_k} \right] \delta_{ij} 
	+ \alpha(\rho,\theta) \dfrac{\partial\rho}{\partial x_i} \dfrac{\partial\rho}{\partial x_j},
\end{equation*}
However, if it were assumed: 
\begin{equation}
	\alpha(\rho, \theta) = - \rho A(\rho), 
\end{equation}
the same structure of equilibrium stress tensor would have been obtained as in \eqref{tenzor napona 2_Heida-Malek}. On the other hand, assumption that material function $\alpha$ does not depend on temperature $\theta$ implies the loss of generality of the nonequilibrium fluxes \eqref{Konstitutivne relacije Kortevegovi fluidi-1} and \eqref{Konstitutivne relacije Kortevegovi fluidi-2}, where $\partial \alpha/\partial \theta$ explicitly appears. 

\subsection{On the boundary conditions} 

The boundary conditions of the theory deserve comment. The entropy flux \eqref{protok entropije Kortevegovi fluidi} determines natural boundary conditions through its normal component at the boundary: $\varphi_j n_j = \frac{1}{\theta}(q_j + \rho\alpha\,\rho_{,j}\,\partial_k v_k) n_j$, which is the thermodynamic counterpart of the variational natural boundary conditions in the free-energy functional approach of \cite{anderson1998diffuse}. In particular, the term $\alpha\,\rho_{,j}\,n_j$ at the boundary controls the contact angle in wetting problems. 

Boundary conditions in the diffuse interface models were particularly studied by \cite{heida2013derivation,heida2014thermodynamics}, where the emphasis was on thermodynamic consistency. In general, there are different approaches that could lead to thermodynamically consistent boundary conditions, and this will be the subject of a separate study. 

%%%%%% 

\section{Conclusions} \label{Sec:Conclusions} 

%%%%%% 

In this study, the constitutive modelling of Korteweg fluids is analyzed from the point of view of Liu's method of multipliers. For the sake of completeness, the main algebraic lemma (Lemma \ref{lemma:Liu-Lemma}) is given at the beginning. Afterwards, the Liu method is applied to the Euler fluids in Section \ref{Sec:EulerFluids}. Although no new results are provided in this part, certain novelties are emphasized in the approach, which are applied to Korteweg fluids in the sequel. 

Main results of the paper are given in Section \ref{Sec:KortewegFluids}. Features that make our results appealing and different from previous studies are the following: 
\begin{enumerate}
	\item Constitutive assumptions facilitated inclusion of the capillary effects into specific entropy, and simpler derivation of the multipliers. 
	\item Retaining the skew-symmetric part of the velocity gradient enabled derivation of the additional multiplier using symmetry arguments. 
	\item Entropy flux is determined from the residual inequality, instead of solving highly involved system of partial differential equations, or by assuming its dependence on the elements of the constitutive state space. 
	\item Korteweg stresses are derived from the equilibrium conditions---vanishing of the entropy production and its minimization in equilibrium. In such a way, their equilibrium character is emphasized. 
	\item Nonequilibrium parts of thermodynamic fluxes are determined in the most general form admitted by the structure of constitutive state space. 
	\item The generalized Gibbs' relation is derived, which inherits the capillary effects. 
	\item The material function $\alpha(\rho,\theta)$ is allowed to depend on temperature, which is consistent with kinetic-theory results of \cite{bhattacharjee2024temperature} and leads to cross-coupling between thermal and mechanical effects in the non-equilibrium constitutive relations.
	\item The non-uniqueness of the equilibrium stress tensor is identified and its resolution through the symmetry of the constraint argument is discussed.
\end{enumerate}

A question may be raised about rationale for this study when other approaches may yield similar outcomes. The first reason lies in the fact that Liu's method requires minimum number of \emph{ad hoc} physical assumptions. Main influence of the physical reasoning comes through the choice of governing equations and the constitutive state space. 

The second reason is related to possible further applications to the mixture of Korteweg fluids. In this case, classical approaches deeply rely on physical assumptions, which influence the closure procedure and dictate the structure of constitutive relations. On the other hand, Liu's method facilitates mixture modelling in a straightforward manner if the same physical assumptions (the constitutive state spaces) are applied to all the mixture constituents (\cite{ruggeri2007hyperbolic}). This is assumed to be the starting point of our forthcoming work.

Further directions include the extension to higher-grade fluids, like in \cite{gorgone2021characterization,morro2023korteweg,paolucci2022second}, where the complete execution of Liu's procedure (explicit multiplier computation, entropy flux derivation, linearization) remains an open problem. From the experimental side, the two non-equilibrium cases (eqs.\ \eqref{Konstitutivne relacije Kortevegovi fluidi-1}--\eqref{Konstitutivne relacije Kortevegovi fluidi-2}) predict different dispersion relations for capillary waves at diffuse interfaces, which could be tested against molecular dynamics simulations or measurements near the critical point (\cite{bhattacharjee2025thermodynamically}). Finally, the classical holographic property identified in Section~\ref{Sec:EnergyStructure} connects the present framework to the thermodynamic derivation of quantum-mechanical fluid equations (\cite{szucs2025thermodynamic,van2023holographic}), suggesting that Liu's method may provide a unified treatment of capillary, gravitational and quantum phenomena in weakly nonlocal fluids.

%%%%%%
\section{Acknowledgement}

The authors acknowledge networking support by the grant   NKFIH NKKP-Advanced 150038 and COSTAction FuSe CA24101. This article/publication is based upon work from COST Action    FuSe, CA24101, supported by COST (European Cooperation in Science and Technology).

\appendix 

%%%%%% 

\section{Useful identities} 

\subsection{Velocity gradient} 

Velocity gradient $L_{ij}$ is decomposed into symmetric part $D_{ij}$ and skew-symmetric part $W_{ij}$. Along with divergence of velocity, they read  
%\begin{equation} 
\begin{align} \label{App:Lij} 
	\dfrac{\partial v_i}{\partial x_j} = L_{ij} & = D_{ij} + W_{ij},
	\nonumber \\ 
	D_{ij} = \dfrac{1}{2}\left(\dfrac{\partial v_i}{\partial x_j} 
	+ \dfrac{\partial v_j}{\partial x_i}\right), & \quad
	W_{ij} = \dfrac{1}{2}\left(\dfrac{\partial v_i}{\partial x_j} 
	- \dfrac{\partial v_j}{\partial x_i}\right), 
	\\ 
	\dfrac{\partial v_j}{\partial x_j} = D_{jj} & = D_{jk}\delta_{jk} = D_{ij}\delta_{ij}. 
	\nonumber 
\end{align}
%\end{equation}
Derivation of the balance for density gradient requires computation of divergence of the velocity gradient. The following identities are useful 
%\begin{equation} 
\begin{align} \label{App:divergenceLij}
	\dfrac{\partial D_{ji}}{\partial x_j} 
	& = \dfrac{1}{2}\left(\dfrac{\partial^2 v_j}{\partial x_i\partial x_j} + 
	\dfrac{\partial^2 v_i}{\partial x_j\partial x_j}\right), 
	\nonumber \\
	\dfrac{\partial W_{ji}}{\partial x_j} 
	& = \dfrac{1}{2}\left(\dfrac{\partial^2 v_j}{\partial x_i\partial x_j} 
	- \dfrac{\partial^2 v_i}{\partial x_j\partial x_j}\right),
	\\ 
	\dfrac{\partial D_{ji}}{\partial x_j} + \dfrac{\partial W_{ji}}{\partial x_j} 
	& = \dfrac{\partial^2 v_j}{\partial x_i \partial x_j} 
	= \dfrac{\partial}{\partial x_i}\dfrac{\partial v_j}{\partial x_j} 
	\nonumber \\
	& = \dfrac{\partial}{\partial x_i}(D_{jk}\delta_{jk}) 
	= \dfrac{\partial D_{jk}}{\partial x_i}\delta_{jk}. 
	\nonumber 
\end{align}
%\end{equation}

\subsection{Euler fluids} 

In the analysis of entropy balance law for Euler fluids, the following identities will be useful. The first one is the rate of change of specific entropy, where we used assumptions \eqref{konstitutivne relacije s-phi Ojler}$_1$ and \eqref{konstitutivne relacije s hipoteza 1 Ojlerovi fluidi} about constitutive relations
%\begin{equation} 
\begin{align} \label{App:specificna entropija Ojler}
	\dfrac{\partial s}{\partial t}&=\dfrac{\partial s}{\partial e}\dfrac{\partial e}{\partial t}+\dfrac{\partial s}{\partial\rho}\dfrac{\partial\rho}{\partial t}+\dfrac{\partial s}{\partial\theta}\dfrac{\partial\theta}{\partial t}
	\nonumber \\
	&=\dfrac{\partial s}{\partial e}\left(\dfrac{\partial e}{\partial\rho}\dfrac{\partial\rho}{\partial t}+\dfrac{\partial e}{\partial\theta}\dfrac{\partial\theta}{\partial t}\right)+\dfrac{\partial s}{\partial\rho}\dfrac{\partial\rho}{\partial t}+\dfrac{\partial s}{\partial\theta}\dfrac{\partial\theta}{\partial t} 
	\\
	&=\left(\dfrac{\partial s}{\partial e}\dfrac{\partial e}{\partial\rho}+\dfrac{\partial s}{\partial\rho}\right)\dfrac{\partial\rho}{\partial t}+\left(\dfrac{\partial s}{\partial e}\dfrac{\partial e}{\partial\theta}+\dfrac{\partial s}{\partial\theta}\right)\dfrac{\partial\theta}{\partial t} 
	\nonumber \\
	&=\overline{\dfrac{\partial s}{\partial\rho}}\dfrac{\partial\rho}{\partial t}+\overline{\dfrac{\partial s}{\partial\theta}}\dfrac{\partial\theta}{\partial t}. 
	\nonumber 
\end{align}
%\end{equation}
Similarly, using \eqref{konstitutivne relacije s-phi Ojler}$_2$, divergence of the entropy flux is computed 
\begin{equation} \label{App:protok entropije Ojlerovi fluidi}
	\dfrac{\partial\varphi_j}{\partial x_j} 
	= \dfrac{\partial\varphi_j}{\partial\rho} \dfrac{\partial\rho}{\partial x_j} 
	+ \dfrac{\partial\varphi_j}{\partial\theta}\dfrac{\partial\theta}{\partial x_j}. 
\end{equation} 
Finally, taking into account \eqref{KonstitutivneRelacije-Ojler2}, divergence of the stress tensor and its product with velocity gradient is computed 
%\begin{equation} 
\begin{align} \label{App:tenzor napona Ojler}
	\dfrac{\partial t_{ij}}{\partial x_j} & = 
	- \dfrac{\partial p}{\partial x_j}\delta_{ij} 
	= - \dfrac{\partial p}{\partial x_i} 
	= - \left(\dfrac{\partial p}{\partial\rho}\dfrac{\partial\rho}{\partial x_i} 
	+ \dfrac{\partial p}{\partial\theta}\dfrac{\partial\theta}{\partial x_i} \right), 
	\\
	t_{ij}\dfrac{\partial v_i}{\partial x_j} & 
	= - p \delta_{ij} \dfrac{\partial v_i}{\partial x_j} 
	= - p \dfrac{\partial v_j}{\partial x_j} 
	= - p D_{ij} \delta_{ij}. 
	\nonumber 
\end{align} 
%\end{equation}

\begin{comment}
\section{}\label{}

% Numbered list
% Use the style of numbering in square brackets.
% If nothing is used, default style will be taken.
%\begin{enumerate}[a)]
%\item 
%\item 
%\item 
%\end{enumerate}  

% Unnumbered list
%\begin{itemize}
%\item 
%\item 
%\item 
%\end{itemize}  

% Description list
%\begin{description}
%\item[]
%\item[] 
%\item[] 
%\end{description}  

\clearpage %%Remove this from your manuscript

% Figure
\begin{figure}%[]
  \centering
%    \includegraphics{}
    \caption{}\label{fig1}
\end{figure}

\begin{table}%[]
\caption{}\label{tbl1}
\begin{tabular*}{\tblwidth}{@{}LL@{}}
\toprule
  &  \\ % Table header row
\midrule
 & \\
 & \\
 & \\
 & \\
\bottomrule
\end{tabular*}
\end{table}

% Uncomment and use as the case may be
%\begin{theorem} 
%\end{theorem}

% Uncomment and use as the case may be
%\begin{lemma} 
%\end{lemma}

%% The Appendices part is started with the command \appendix;
%% appendix sections are then done as normal sections
%% \appendix

\section{}\label{}
\end{comment}

% To print the credit authorship contribution details
\printcredits

%% Loading bibliography style file
%\bibliographystyle{model1-num-names}
\bibliographystyle{cas-model2-names.bst}

% Loading bibliography database
\bibliography{Korteweg-Bibliography.bib}

% Biography
%\bio{}
% Here goes the biography details.
%\endbio

%\bio{pic1}
% Here goes the biography details.
%\endbio

\end{document}